\documentclass[12pt,reqno]{amsart}
\usepackage{amsthm,amsfonts,amssymb,euscript}

\newcommand{\bea}{\begin{eqnarray}}
\newcommand{\eea}{\end{eqnarray}}
\def\beaa{\begin{eqnarray*}}
\def\eeaa{\end{eqnarray*}}
\def\ba{\begin{array}}
\def\ea{\end{array}}
\def\be#1{\begin{equation} \label{#1}}
\def \eeq{\end{equation}}

\def\a{{\alpha}}

\def\b{{\beta}}
\def\be{{\beta}}

\def\de{\delta}

\def\ep{\epsilon}
\def\eps{\epsilon}

\def\la{\lambda}

\def\Si{\Sigma}

\def\Om{\Omega}

\def\al{\alpha}

\def\c{\cdot}
\def\rh{{\rho}}

\def\AA{{\mathcal A}}
\def\Aa{{\mathcal A}}
\def\BB{{\mathcal B}}
\def\Bb{{\mathcal B}}
\def\Cc{{\mathcal C}}
\def\MM{{\mathcal M}}

\def\EE{{\mathcal E}}
\def\HH{{\mathcal H}}

\def\SS{{\mathcal S}}
\def\Ss{{\mathcal S}}

\def\RR{{\mathbb R}}

\def\AA{{\mathcal A}}

\def\D{{\bf D}}

\def\T{{\bf T}}

\def\g{{\bf g}}

\def\RRR{{\mathbb R}}

\def\f12{{\frac 1 2}}

\def\f{\widetilde{f}}

\newtheorem{theorem}{Theorem}[section]
\newtheorem{lemma}[theorem]{Lemma}
\newtheorem{proposition}[theorem]{Proposition}

\newtheorem{definition}[theorem]{Definition}
\newtheorem{remark}[theorem]{Remark}

\numberwithin{equation}{section}

\begin{document}

\title[Uniqueness results for ill posed characteristic problems]
{Uniqueness results for ill posed characteristic problems in curved space-times}
\author{Alexandru D. Ionescu}
\address{University of Wisconsin -- Madison}
\email{ionescu@math.wisc.edu}
\author{Sergiu Klainerman}
\address{Princeton University}
\email{seri@math.princeton.edu}
\begin{abstract}
We prove two uniqueness theorems concerning 
linear wave equations; the first theorem is
in Minkowski space-times, while the second is in the domain of outer communication of a Kerr black hole. Both theorems concern
 ill-posed  Cauchy problems  on bifurcate,  characteristic 
hypersurfaces. In  the case of the Kerr space-time, the hypersurface is precisely the event horizon of the black hole. The uniqueness theorem in this case, based on two
Carleman estimates,  is intimately connected to our strategy to prove uniqueness of the Kerr black holes among smooth, stationary solutions of the Einstein-vacuum equations, as formulated in 
\cite{Ion-K}.
\end{abstract}

\maketitle

\tableofcontents
\section{Introduction}\label{section1}

The goal of the paper is to prove two uniqueness results for the  Cauchy  problem in the 
exterior of a  bifurcate characteristic surface. In the  simplest  case  of 
the wave equation in   Minkowski space  $\RRR^{1+d}$,
\beaa
 \square \phi=0,\qquad   \qquad \square=  -\partial_  t^2+\sum_{i=1}^d \partial_{i}^2
 \eeaa
  the  problem  is to find  solutions  
 in the  exterior   domain  $\EE_a=\{(t,x): |x|>|t|+a\}$, $a\ge 0$,   with prescribed data
   on the boundary  $\HH_a=\{(t,x): |t|=|x|+a\}$.
    The problem is known to be \textit{ill posed}, that is, 
   \begin{enumerate}
   \item Solutions may not exist for  smooth, non-analytic, initial conditions.
   \item There is no continuous dependence on the data.
   \end{enumerate}
The situation is similar to the better known case of 
   the Cauchy problem prescribed on a  time-like  characteristic
   hypersurface, such as $x^d=0$.  The Cauchy--Kowalewski theorem
   allows one to solve the problem for analytic initial data, but 
   solutions may not exist in the smooth case. It is known in fact
   that smooth solutions cannot be prescribed freely, since certain
   necessary compatibilities may be violated. 

  Though existence fails, one can often prove uniqueness.
  A general  result  due to Holmgren,  improved by F.  John  \cite{John1},  shows that  the  non-characteristic initial value problem for linear  equations with analytic   coefficients is  locally unique   in the class of smooth 
  solutions, see \cite{John2}.  The case of equations with 
  smooth coefficients is considerably more complicated.
  An important  counterexample to uniqueness was provided by P. Cohen \cite{Cohen}, see also \cite{Ho2} and \cite{A-B} for more general families of examples. Thus, in the case  of the Cauchy problem for a  time-like  hypersurface   (such as  $x^d=0$), 
   even a zero  order,  smooth,   perturbation of the wave operator $\square$  can cause uniqueness  to fail. We  note also  that, 
      there  cannot be, in general (unless one considers solutions with  suitable decay at infinity 
      such as discussed in \cite{KeRuSo}),   unique continuation across   characteristic hyperplanes, see the counterexample and the discussion in \cite[Theorem 8.6.7]{Ho}.  On the other hand,  there exist conditions   which can guarantee uniqueness,  most importantly those of H\"{o}rmander \cite[Chapter 28]{Ho}. See also \cite{RoZu},  \cite{Ta} and the references therein for uniqueness results under partial analyticity assumptions. These results prove uniqueness for a large class of problems which include, in particular, 
       the Cauchy problem  on an arbitrary,  non-characteristic, 
     time-like hypersurface for the wave equation 
     $ \square_\g\phi =0$,
       corresponding to
    a time independent Lorentz metric of the form
    $
    -\g_{00}(x) dt^2+\g_{ij}(x) dx^i dx^j$  with $ \g_{00}>0$
    and $(\g_{ij})_{i,j=1}^d$ positive definite.     
    The  method  of proof for  these and other  modern  unique continuation results is based on Carleman  type  estimates.

  The case of  ill posed problems  for 
   bifurcate characteristic hypersurfaces, i.e. surfaces composed of two characteristic hypersurfaces which  intersect transversally, seems to have been
     first studied  by Friedlander\footnote{In  \cite{Fried3} he also considers  a similar, ill posed,  characteristic problem at infinity, 
     concerning  uniqueness of solutions with identical  radiation fields.} \cite{Fried1}, 
     by using a variation  of Holmgren's method of proof. The  same problem for  equations with smooth coefficients, seems not to      have been  specifically   considered in the literature.
    Yet   it is precisely this case  which seems to be of  considerable   importance 
     in General Relativity, particularly for the   problem of uniqueness of stationary, smooth  solutions of the Einstein field equations,  see discussion  in  \cite{Ion-K}.  
     Indeed, it turns out that  remarkable simplifications occur   for the geometry
     of  bifurcate horizons for general, stationary, asymptotically flat  black hole solutions of the Einstein-vacuum equations, verifying reasonable regularity assumptions.  For such regular black hole space-times, Hawking has shown, see \cite{H-E},  then   there must exist an additional  Killing 
     vector-field defined on the event horizon, tangent to the generators of the horizon.  In the case when the space-time is real  analytic one can extend this additional Killing vector-field to the entire exterior region, and deduce that the space-time  must be not only 
      stationary but also  axially symmetric.    A satisfactory  uniqueness result (due to
      Carter \cite{Ca1} and Robinson \cite{Rob}) 
       is known for  stationary  solutions  which have this additional symmetry. However, in the  smooth, non-analytic case, the problem of extending Hawking's Killing vector-field from the horizon to the exterior region leads to an  ill posed   characteristic problem. This appears to be the key obstruction to proving the analogue of Hawking's uniqueness theorem in the class of smooth, non-analytic space-times.

     Motivated by this latter problem, to avoid the analyticity assumption  we are proposing a completely different approach\footnote{See the longer discussion in \cite{Ion-K}.} based on the following
 ingredients.
 \begin{enumerate} 
 \item The    Kerr space-times  can be locally characterized,   among stationary solutions,  by the vanishing of   a four covariant tensor-field, 
 called the  Mars-Simon tensor $\Ss$ introduced in \cite{Ma1}. 
 \item The Mars-Simon tensor-field  $\Ss$ verifies a covariant system of  wave equation of the form (see also first equation in  \eqref{va4}) ,
 \begin{equation}
 \square_\g  \Ss= \AA\cdot \D \Ss+ \BB\cdot \Ss. \label{eq1}
 \end{equation}
 Moreover, since $\g$ is stationary, we know that there exists a globally defined
 Killing vector-field $\xi$, which is time-like at space-like infinity. Thus it is easy to verify that  the Lie derivative of $\Ss$ with respect to $\xi$  vanishes identically.
  \begin{equation}
  \mathcal L_\xi\Ss=0.\label{eq2}
  \end{equation}
 \item One can show that the tensor-field  $\Ss$ vanishes identically   on the bifurcate  horizon  $\HH$  of the stationary metric $g$. We show this by making an assumption
  (automatically satisfied on a Kerr metric)
 concerning  the vanishing of a   complex scalar  on the bifurcate sphere of the horizon.
 \item  Using a first Carleman estimate for the covariant wave equation
 \eqref{eq1}  we show that $\Ss$ vanishes in  a neighborhood of  the bifurcate
 sphere. This step  does not require condition \eqref{eq2}, indeed it is a result
 that applies to general equation of type \eqref{eq1} in a neighborhood of a 
 regular    bifurcate characteristic   hypersurface,  
  for a general  Lorentz  metric $\g$.
  \item To extend the vanishing of $\Ss$ to the entire domain of outer communication we need a more sophisticated 
  Carleman estimate which depends  in an essential fashion,  among other considerations,    on equation  \eqref{eq2}. 
 \end{enumerate} 

In this paper we prove, see Theorem \ref{Main1kerr}, a  global uniqueness     result 
     for   tensor-field  solutions to  covariant equations of the form \eqref{eq1} and \eqref{eq2} on   the domain of outer communication of a Kerr background,  which vanish on the event horizon. The condition \eqref{eq2} relative to the  stationary Killing 
  vector-field  $\xi$, which is important to prove a global result, is  justified by the fact  that  the problem of uniqueness of Kerr is restricted, naturally, to stationary solutions of the Einstein vacuum
     equations (see discussion in  \cite{Ion-K}). We also discuss a simple model problem, see Theorem \ref{Main1}, concerning  scalar  linear  wave equations in the exterior domain $\EE=\EE_1$ of the Minkowski space-time  with prescribed  data  on the characteristic boundary $\HH=\HH_1$.

  We would like to thank  A. Rendall  for bringing to our attention 
  the work of Friedlander, \cite{Fried1}, \cite{Fried3}.

\subsection{A model problem  in Minkowski spaces}\label{minmainthm}

Assume $d\geq 1$ and let $(\mathcal{M}=\RR\times \RR^d,{\mathbf{m}})$ denote the usual Minkowski space of dimension $d+1$. We define the subsets of $\MM$
\begin{equation}\label{lo1}
\mathcal{E}=\{(t,x)\in\MM:\,|x|>|t|+1\},
\end{equation}
and
\begin{equation}\label{lo2}
\HH=\delta(\EE)=\{(t,x)\in\MM:\,|x|=|t|+1\}.
\end{equation}
Let $\overline{\EE}=\EE\cup\HH$. Our first theorem concerns a uniqueness property of solutions of wave equations on $\mathcal{E}$.

\begin{theorem}\label{Main1}
Assume $\phi\in C^2(\mathcal{M})$, $A,B^l\in C^0(\mathcal{M})$, $l=0,\ldots,d$, and
\begin{equation}\label{lamin2}
\square \phi=A\cdot \phi+\sum_{l=0}^dB^l\cdot \partial_l\phi\,\,\text{ on }\,\,\mathcal{E}.
\end{equation}
Assume that $\phi\equiv0\text{ on }\mathcal{H}$. Then $\phi\equiv 0$ on $\overline{\EE}$.
\end{theorem}

Theorem \ref{Main1} extends easily to diagonal systems of scalar equations. We remark that in Theorem \ref{Main1} we do not assume any global bounds on the coefficients $A$ and $B^l$.  Also, we make no assumption on the vanishing of the derivatives of $\phi$ on $\mathcal{H}$, which is somewhat surprising given that $\square$ is a second order operator. This is possible because of the special bifurcate characteristic structure of the surface $\HH$.  

The proof of Theorem \ref{Main1}, which is given in section \ref{MinProof}, follows from a standard Carleman inequality with a suitably defined pseudo-convex weight. However, the simple statement of Theorem \ref{Main1} appears to be new. We include it here mostly as a model  result  to illustrate, in a very simple case, the connection between bifurcate characteristic horizons and unique continuation properties of solutions of wave equations. 

\subsection{The main theorem in the Kerr spaces}\label{kerrmainthm}

Let $(\mathbf{K}^4,\mathbf{g})$ denote the maximally extended Kerr spacetime of mass $m$ and angular momentum $ma$ (see the appendix for some details and explicit formulas). We assume 
\begin{equation*}
m>0\text{ and }a\in[0,m).
\end{equation*}
Let $\mathbf{E}^4$ denote a domain of outer communication of $\mathbf{K}^4$, and $\HH=\delta(\mathbf{E}^4)$ the corresponding event horizon. Let $\mathbf{M}^4$ denote an open neighborhood of $\mathbf{E}^4\cup\HH$ in $\mathbf{K}^4$, and let $\xi$ denote a Killing vector field on $\mathbf{E}^4$ which is timelike at the spacelike infinity in $\mathbf{E}^4$. Let $\mathbb{T}(\mathbf{M}^4)$ denote the space of smooth vector-fields on ${\mathbf{M}}^4$, and let $\mathbb{T}_s^r({\mathbf{M}}^4)$, $r,s\in\mathbb{Z}_+$, denote the space of complex-valued tensor-fields of type $(r,s)$ on ${\mathbf{M}}^4$. Our main theorem concerns a uniqueness property of certain solutions of covariant wave equations on $\mathbf{E}^4$.

\begin{theorem}\label{Main1kerr} 
Assume $k\in\mathbb{Z}_+$, $\Ss\in\mathbb{T}_k^0({\mathbf{M}}^4)$, $\Aa\in\mathbb{T}_k^k({\mathbf{M}}^4)$, $\Bb\in\mathbb{T}_k^{k+1}({\mathbf{M}}^4)$, $\Cc\in\mathbb{T}_k^k({\mathbf{M}}^4)$, and
\begin{equation}\label{va4}
\begin{cases}
&\square_\mathbf{g}\Ss_{\al_1\ldots \al_k}=\Ss_{\be_1\ldots\be_k}{\Aa^{\be_1\ldots\be_k}}_{\al_1\ldots\al_k}+\D_{\be_{k+1}}\Ss_{\be_1\ldots\be_k}{\Bb^{\be_1\ldots\be_{k+1}}}_{\al_1\ldots\al_k};\\
&\mathcal{L}_{\xi}\Ss_{\al_1\ldots\al_k}=\Ss_{\be_1\ldots\be_k}{\Cc^{\be_1\ldots\be_k}}_{\al_1\ldots\al_k},
\end{cases}
\end{equation}
in $\mathbf{E}^4$. Assume in addition that
$
\Ss\equiv 0\text{ on }\HH.
$
Then,
$
\Ss\equiv 0\text{ on }\mathbf{E}^4\cup\HH.
$
\end{theorem}

\section{Unique continuation and conditional Carleman inequalities}\label{generalCarleman}
\subsection{General considerations}

Our proof of Theorem \ref{Main1kerr} is based on a global unique continuation strategy.  We say
that a linear differential operator $L$, in a domain $\Omega\subset\RRR^d$,   satisfies the unique continuation property with respect  to a smooth, oriented, hypersurface $\Si\subset\Omega$,  if  any smooth solution  of $L\phi=0$ which vanishes on one side of $\Si$  must in fact vanish in a small neighborhood of $\Si$. Such a property depends,
of course, on the interplay between the properties of the operator $L$ and the hypersurface $\Sigma$. A classical result of H\"{o}rmander, see for example Chapter 28 in \cite{Ho}, provides sufficient conditions for a scalar linear equation which guarantee that the unique continuation property holds. In the particular case of
the scalar wave equation, 
$
\square_\g\phi=0,
$ 
and a smooth surface $\Sigma$ defined by the equation $h=0$, $\nabla h\neq 0$, H\"{o}rmander's pseudo-convexity condition takes  the form,
\begin{equation}\label{HoCond}
\D^2h(X,X)<0\qquad\text{ if }\qquad \g(X,X)=\g(X,\D h)=0
\end{equation}
at all points on the surface $\Sigma$, where we assume that $\phi$ is known to vanish on the side of $\Sigma$ corresponding to $h<0$. 

In our situation, we plan  to apply the general philosophy of unique continuation to  the covariant wave equation (see the first equation in \eqref{va4}),
\begin{equation}\label{cartoon1}
\square_\g\Ss=\AA\ast\SS+\BB\ast\D\Ss.
\end{equation}
We know that $\Ss$ vanishes on the horizon $\mathcal{H}$ and we would like to prove, by unique continuation, that $\Ss$ vanishes in the entire domain of outer communication.   In implementing such a strategy one encounters the following
difficulties:
\begin{enumerate}
\item The  horizon $\mathcal{H}=\mathcal{H}^+\cup\mathcal{H}^-$ is characteristic and not
smooth in a neighborhood of the bifurcate sphere.
\item Even though one can show that an appropriate variant of H\"{o}rmander's pseudo-convexity condition holds true  along  the horizon, in a neighborhood of the bifurcate
sphere,  such a condition may fail to be true  slightly away from the horizon,  within the ergosphere  region of the stationary space-time  where $\xi$ is space-like. 
\end{enumerate}
Problem (1) can be dealt with by exploiting the fact that the horizon is a bifurcate characteristic hypersurface, which, in particular, is sufficient to allow us to prove that higher order derivatives of $\Ss$ vanish on the horizon. Problem (2) is more serious, in the case when $a$ is not small compared to $m$, because of the existence of null geodesics  trapped within the ergoregion $m+\sqrt{m^2-a^2}\le r\le  m+\sqrt{m^2-a^2\cos^2\theta}$.  Indeed    
   surfaces  of the form $r\Delta=m(r^2-a^2)^{1/2}$,   which intersect the ergoregion for $a$ sufficiently close to $m$,  are known to contain such    null geodesics, see \cite{Ch}.  One can  show that the  presence of  trapped  null geodesics invalidates H\"ormander's pseudo-convexity condition. Thus, even in the case of the scalar wave equation  $\square_\g\phi=0$ in  such a Kerr metric, one cannot guarantee,  by a classical unique continuation argument (in the absence of additional conditions)  that  $\phi$ vanishes beyond a small neighborhood of the horizon. 
   
In order to overcome this main difficulty  we need to exploit the second identity in \eqref{va4}, namely
\begin{equation}\label{cartoon2}
\mathcal{L}_\T\Ss=\Cc\ast\Ss.
\end{equation}
Observe that \eqref{cartoon2} can, in principle, transform  \eqref{cartoon1}
into a much simpler elliptic problem, in  any domain which lies strictly outside the
ergoregion (where $\xi$ is strictly  time-like). Unfortunately this   possible strategy  is not available to us when $a$ is not small compared to $m$, since,
as we have remarked above,   we cannot hope  to  extend the vanishing  of $\Ss$,
by a simple  analogue of H\"ormander's pseudo-convexity condition,  
beyond the   first trapped null geodesics. 

Our solution is to extend H\"{o}rmander's classical  pseudo-convexity condition \eqref{HoCond} to one which takes into account both equations
\eqref{cartoon1} and \eqref{cartoon2} simultaneously.  These considerations lead  to the following
 qualitative, $\xi$-conditional, pseudo-convexity condition, 
\begin{equation}\label{HoCond2}
\begin{split}
&\xi(h)=0;\\
&\D^2h(X,X)<0\qquad\text{ if }\qquad \g(X,X)=\g(X,\D h)=\g(\xi,X)=0.
\end{split}
\end{equation}
We will show that this condition can be verified  in  all Kerr spaces  $a\in[0,m)$, for the simple function $h=r$, where $r$ is one  of the Boyer--Lindquist coordinates. Thus \eqref{HoCond2}   is a good substitute for  the more general condition \eqref{HoCond}. The fact that the two geometric identities \eqref{cartoon1} and \eqref{cartoon2} cooperate exactly in the right way, via   \eqref{HoCond2},   thus  allowing us  to compensate for  both the failure of 
 condition \eqref{HoCond}   as well as  the failure of the vector field $\xi$ to be time-like in the ergoregion, seems to us to be a very remarkable property of the Kerr spaces. In the next subsection we give a quantitative version    of the condition and state a Carleman estimate of sufficient generality to cover    all our needs.

\subsection{A conditional Carleman inequality of sufficient generality}
Unique continuation properties are often proved using Carleman inequalities. In this subsection we state a sufficiently general Carleman inequality, Proposition \ref{Cargen}, under a quantitative conditional pseudo-convexity assumption. This general  Carleman inequality is used to show first that $\Ss$ vanishes in a small neighborhood of the bifurcate sphere $S_0$ in $\overline{\mathbf{E}^4}$, using only the first identity in \eqref{va4}, and then to prove that $\Ss$ vanishes in the entire exterior domain using both identities in \eqref{va4}. The two applications are genuinely different, since, in particular, the horizon is a bifurcate surface which is not smooth and the weights needed in this case have to be ``singular'' in an appropriate sense. In order to be able to cover both applications and prove unique continuation in a quantitative sense, we work with a more  technical notion of conditional pseudo-convexity than  \eqref{HoCond2}, see Definition \ref{psconvex} below.      

Let $B_r=\{x\in\mathbb{R}^4:|x|<r\}$ denote the standard open ball in $\mathbb{R}^4$. Assume that $(M,\g)$ is a smooth Lorentzian manifold of dimension $4$, $x_0\in M$, and $\Phi^{x_0}:B_1\to B_1(x_0)$ is a coordinate chart. For  simplicity  of notation, let $B_r(x_0)=\Phi^{x_0}(B_r)$, $r\in(0,1]$. For any smooth function $\phi:B\to\mathbb{C}$, where $B\subseteq B_1(x_0)$ is an open set, and $j=0,1,\ldots$ let
\begin{equation}\label{gradco}
|D^j\phi(x)|=\sum_{\al_1,\ldots,\al_j=1}^4|\partial_{\al_1}\ldots\partial_{\al_j}\phi(x)|.
\end{equation}
Let $\g_{\al\be}=\g(\partial_\al,\partial_\be)$ and assume that $V=V^\alpha\partial_\alpha$ is a vector-field on $B_1(x_0)$. We assume that
\begin{equation}\label{po1}
\sup_{x\in B_1(x_0)}\sum_{j=0}^6\sum_{\al,\beta=1}^4|D^j\g_{\al\be}|+|D^j\g^{\al\be}|+|D^jV^\beta|\leq A_0.
\end{equation}
In our applications $V=0$ or $V=\xi$.

\begin{definition}\label{psconvex}
A family of weights $h_\eps:B_{\eps^{10}}(x_0)\to\mathbb{R}_+$, $\eps\in(0,\eps_1)$, $\eps_1\leq A_0^{-1}$, will be called $V$-conditional pseudo-convex if for any $\eps\in(0,\eps_1)$ \begin{equation}\label{po5}
\begin{split}
h_\eps(x_0)=\eps,\quad\sup_{x\in B_{\eps^{10}}(x_0)}\sum_{j=1}^4\eps^j|D^jh_\eps(x)|\leq\eps/\eps_1,\quad |V(h_\eps)(x_0)|\leq\eps^{10},
\end{split}
\end{equation}
\begin{equation}\label{po3.2}
\D^\alpha h_\eps(x_0)\D^\be h_\eps(x_0)(\D_\al h_\eps\D_\be h_\eps-\eps\D_\al\D_\be h_\eps)(x_0)\geq\eps_1^2,
\end{equation}
and there is $\mu\in[-\eps_1^{-1},\eps_1^{-1}]$ such that for all vectors $X=X^\alpha\partial_\alpha\in\mathbb{T}_{x_0}(M)$
\begin{equation}\label{po3}
\begin{split}
&\eps_1^2[(X^1)^2+(X^2)^2+(X^3)^2+(X^4)^2]\\
&\leq X^\al X^\be(\mu\g_{\al\be}-\D_\al\D_\be h_\eps)(x_0)+\eps^{-2}(|X^\al V_\al(x_0)|^2+|X^\al\D_\al h_\eps(x_0)|^2).
\end{split}
\end{equation}
A function $e_\eps:B_{\eps^{10}}(x_0)\to\mathbb{R}$ will be called a negligible perturbation if
\begin{equation}\label{smallweight}
\sup_{x\in B_{\eps^{10}}(x_0)}|D^je_\eps(x)|\leq\eps^{10}\qquad\text{ for  }j=0,\ldots,4.
\end{equation}
\end{definition}

\begin{remark} One can see that the technical conditions \eqref{po5}, \eqref{po3.2}, and \eqref{po3} are related to the qualitative condition \eqref{HoCond2}, at least when $h_\eps=h+\eps$ for some smooth function $h$. The assumption $|V(h_\eps)(x_0)|\leq\eps^{10}$ is a quantitative version of $V(h)=0$. The assumption \eqref{po3} is a quantitative version of the inequality in the second line of \eqref{HoCond2}, in view of the large factor $\eps^{-2}$ on the terms $|X^\al V_\al(x_0)|^2$ and $|X^\al\D_\al h_\eps(x_0)|^2$, and the freedom to choose $\mu$ in a large range. The assumption \eqref{po3.2} is a quantitative version of the condition $\nabla h\neq 0$ (assuming that \eqref{po3} already holds). 

It is important that the Carleman estimates we prove are stable under small perturbations of the weight, in order to be able to use them to prove unique continuation. We quantify this stability in \eqref{smallweight}.  
\end{remark}

We observe that if $\{h_\eps\}_{\eps\in(0,\eps_1)}$ is a $V$-conditional pseudo-convex family, and $e_\eps$ is a negligible perturbation for any $\eps\in(0,\eps_1]$, then
\begin{equation*}
h_\eps+e_\eps\in[\eps/2,2\eps]\text{ in }B_{\eps^{10}}(x_0).
\end{equation*}
The pseudo-convexity conditions of Definition \ref{psconvex} are probably not as general as possible, but are suitable for our applications both in Proposition \ref{Car}, with ``singular'' weights $h_\eps$ and $V=0$, and Proposition \ref{Carl2}, with ``smooth'' weights $h_\ep$ and $V=\xi$.

\begin{proposition}\label{Cargen}
Assume $\eps_1\leq A_0^{-1}$, $\{h_\eps\}_{\eps\in(0,\eps_1)}$ is a $V$-conditional pseudo-convex family, and $e_\eps$ is a negligible perturbation for any $\eps\in(0,\eps_1]$, see Definition \ref{psconvex}. Then there is $\eps\in (0,\eps_1)$ sufficiently small and $C_\eps$ sufficiently large such that for any $\lambda\geq C_\eps$ and any $\phi\in C^\infty_0(B_{\eps^{10}}(x_0))$
\begin{equation}\label{Car1gen}
\lambda \|e^{-\lambda f_\eps}\phi\|_{L^2}+\|e^{-\lambda f_\eps}|D^1\phi|\,\|_{L^2}\leq C_\eps\lambda^{-1/2}\|e^{-\lambda f_\eps}\,\square_{\g}\phi\|_{L^2}+\eps^{-6}\|e^{-\lambda f_\eps}V(\phi)\|_{L^2},
\end{equation}
where $f_\ep=\ln (h_\eps+e_\eps)$.
\end{proposition}

As mentioned earlier, many Carleman estimates such as \eqref{Car1gen} are known, for the particular case $V=0$. Optimal proofs are usually based on some version of the Fefferman-Phong inequality, as in \cite[Chapter 28]{Ho}. A self-contained, elementary proof of Proposition \ref{Cargen}, using only simple integration by parts arguments is given in \cite[Section 3]{Ion-K} (see also Proposition \ref{Carmin1} in section \ref{MinProof} for  a similar proof in a simpler case). We also note that it is useful to be able to track quantitatively the size of the support of the functions for which Carleman estimates can be applied; in our notation, the value of $\eps$ for which  \eqref{Car1gen} holds depends only on the parameter $\eps_1$.

\section{Proof of Theorem \ref{Main1kerr}}\label{KerProof}

\subsection{The first Carleman inequality in Kerr spaces}\label{CarEst1}

The horizon $\HH$ decomposes as
\begin{equation*}
\HH=\HH^+\cup\HH^-,
\end{equation*}
where $\HH^+$ is the boundary of the black hole region and $\HH^-$ is the boundary of the white hole region. Let $\, S_0=\HH^+\cap\HH^-$ denote the bifurcate sphere. In this section we prove a Carleman estimate for functions supported in a small neighborhood of the bifurcate sphere $\, S_0$.

We first construct two suitable defining functions for the surfaces $\HH^+$ and $\HH^-$.

\begin{lemma}\label{ao}
There is an open set $\mathbf{O}\subseteq{\mathbf{M}}^4$, $\, S_0\subseteq \mathbf{O}$, and smooth functions $u,v:\mathbf{O}\to\mathbb{R}$ with the following properties:

(a) We have
\begin{equation*}
\begin{cases}
&\mathbf{E}^4\cap \mathbf{O}=\{x\in\mathbf{O}:u(x)>0\text{ and }v(x)>0\};\\
&\HH^+\cap\mathbf{O}=\{x\in\mathbf{O}:u(x)=0\};\\
&\HH^-\cap\mathbf{O}=\{x\in\mathbf{O}:v(x)=0\}.
\end{cases}
\end{equation*}
In addition, the set $\{x\in\mathbf{O}:u(x),v(x)\in[0,1/2]\}$ is compact.

(b) With $L_3=g^{\al\be}\partial_\al(u)\partial_\be, L_4=g^{\al\be}\partial_\al(v)\partial_\be\in\mathbb{T}(\mathbf{O})$,
\begin{equation}\label{ao1}
\begin{cases}
&\g(L_3,L_3)=0\text{ on }\HH^+\cap\mathbf{O};\\
&\g(L_4,L_4)=0\text{ on }\HH^-\cap\mathbf{O};\\
&\g(L_3,L_4)>0\text{ on }\, S_0.
\end{cases}
\end{equation}

(c) For any smooth function $\phi:\mathbf{O}\to\mathbb{R}$ with the property that $\phi\equiv 0$ on $\HH^+\cap\mathbf{O}$, there is a smooth function $\phi':\mathbf{O}\to\mathbb{R}$ such that
\begin{equation*}
\phi=\phi'\cdot u\text{ on }\mathbf{O}\cap \mathbf{E}^4.
\end{equation*}
Also, for any smooth function $\phi:\mathbf{O}\to\mathbb{R}$ with the property that
$\phi\equiv 0$ on $\HH^-\cap{\mathbf{O}}$, there is a smooth function $\phi':\mathbf{O}\to\mathbb{R}$ such that
\begin{equation*}
\phi=\phi'\cdot v\text{ on }\mathbf{O}\cap \mathbf{E}^4.
\end{equation*}
\end{lemma}

\begin{proof}[Proof of Lemma \ref{ao}] A  more  precise construction of global optical functions $u,v$ is given in \cite{PrIs}. In our problem we do not need this global construction; for simplicity we construct the functions $u,v$ explicitly, using the Kruskal coordinates of the Kerr space-times. In standard Boyer-Lindquist coordinates $(r,t,\theta,\phi)\in (r_+,\infty)\times\mathbb{R}\times(0,\pi)\times\mathbb{S}^1$, $r_\pm=m\pm(m^2-a^2)^{1/2}$ (see the appendix), the Kerr metric on the dense open subset $\widetilde{\mathbf{E}}^4$ of $\mathbf{E}^4$ is
\begin{equation}\label{la1}
ds^2=-\frac{\rho^2\Delta}{\Sigma^2}(dt)^2+\frac{\Sigma^2(\sin\theta)^2}{\rho^2}(d\phi-\omega dt)^2+\frac{\rh^2}{\Delta}(dr)^2+\rho^2(d\theta)^2,
\end{equation}
where
\begin{equation}\label{la2}
\begin{cases}
&\Delta=r^2+a^2-2mr;\\
&\rho^2=r^2+a^2(\cos\theta)^2;\\
&\Sigma^2=(r^2+a^2)\rho^2+2mra^2(\sin\theta)^2=(r^2+a^2)^2-a^2(\sin\theta)^2\Delta;\\
&\omega=\frac{2amr}{\Sigma^2}.
\end{cases}
\end{equation}
We define the function $r_\ast:(r_+,\infty)\to\mathbb{R}$,
\begin{equation}\label{rg20}
r_\ast=\int\frac{r^2+a^2}{r^2+a^2-2mr}\,dr=r+
\frac{2mr_+}{r_+-r_-}\ln(r-r_+)-\frac{2mr_-}{r_+-r_-}\ln(r-r_-).
\end{equation}
With $c_0=\frac{2mr_+}{r_+-r_-}$, we make the changes of variables
\begin{equation}\label{rg3}
r_\ast=c_0(\ln u+\ln v)\text{ and }t=c_0(\ln u-\ln v),
\end{equation}
where $u,v\in(0,\infty)^2$, so
\begin{equation}\label{rg4}
\begin{cases}
&dr_\ast=c_0(u^{-1}du+v^{-1}dv);\\
&dt=c_0(u^{-1}du-v^{-1}dv).
\end{cases}
\end{equation}
We observe also that $\omega(r_+,\theta)=a/(2mr_+)$. We make the
change of variables
\begin{equation}\label{rg5}
\phi=\phi_\ast +\frac{a}{2mr_+}t=\phi_\ast+\frac{ac_0}{2mr_+}(\ln u-\ln v),
\end{equation}
with
\begin{equation}\label{rg6}
d\phi=d\phi_\ast+\frac{ac_0}{2mr_+}(u^{-1}du-v^{-1}dv).
\end{equation}
In the new coordinates $(u,v,\theta,\phi_\ast)\in (0,\infty)\times(0,\infty)\times (0,\pi)\times\mathbb{S}^1$ the Kerr metric
\eqref{la1} becomes
\begin{equation}\label{rg7}
\begin{split}
ds^2&=-\frac{c_0^2\rho^2\Delta^2a^2(\sin\theta)^2}{u^2v^2\Sigma^2(r^2+a^2)^2}
[v^2(du)^2+u^2(dv)^2]+\frac{2c_0^2\rho^2\Delta}{uv}\Big(\frac{1}{\Sigma^2}
+\frac{1}{(r^2+a^2)^2}\Big)du dv\\
&+\frac{\Sigma^2(\sin\theta)^2}{\rho^2}\Big[d\phi_\ast-\frac{c_0\widetilde{\omega}}{uv}
(v du-u dv)\Big]^2+\rho^2(d\theta)^2.
\end{split}
\end{equation}
where $\widetilde{\omega}=\omega-a/(2mr_+)$.

We restrict to the region
\begin{equation*}
\widetilde{\mathbf{O}}=\{(u,v,\theta,\phi_\ast)\in(-c_1,1)^2\times(0,\pi)\times\mathbb{S}^1\},
\end{equation*}
for some constant $c_1>0$ sufficiently small. We examine the coefficients that appear in the Kerr metric \eqref{rg7}. Since $e^{r_\ast/c_0}=uv$ and $r_-<r_+$ (since $a\in[0,m)$), it follows from \eqref{rg20} that $r$ is a smooth function of $uv$ in $\widetilde{\mathbf{O}}$. Moreover $\Delta/(uv)=(r-r_-)(r-r_+)/(uv)$ and $\widetilde{\omega}/(uv)$ are smooth function of $uv$ in $\widetilde{\mathbf{O}}$. Thus the Kerr metric \eqref{rg7} is smooth in $\widetilde{\mathbf{O}}$, and we identify $\widetilde{\mathbf{O}}$ with the corresponding open subset of the Kerr space. We let $\mathbf{O}$ be any open neighborhood of $\, S_0$ contained  in the closure of $\widetilde{\mathbf{O}}$ in ${\mathbf{M}}^4$ (by adding  in the points corresponding to $\theta\in\{0,\pi\}$). It is easy to see that the coordinate functions  $u,v:{\mathbf{O}}\to(-c_1,1)$ verify the conclusions of the lemma.
\end{proof}

Assume now that $x_0\in \, S_0$, $B_r=\{x\in\mathbb{R}^4:|x|<r\}$, and $\Phi^{x_0}:B_1\to\mathbf{O}$, $\Phi^{x_0}(0)=x_0$, is a smooth coordinate chart around $x_0$. In view of \eqref{ao1}
\begin{equation}\label{gt1}
\delta_0=\inf_{\, S_0}\,\,\g(L_3,L_4)>0.
\end{equation}
In follows from \eqref{ao1} that there is $\eps_0\in(0,1/2]$ such that
\begin{equation}\label{gt2}
\g(L_3,L_4)> \delta_0/2\text{ and }|\g(L_3,L_3)|+|\g(L_4,L_4)|< \delta_0/100\,\text{ on }B_{\eps_0}(x_0),
\end{equation}
where $B_r(x_0)=\Phi^{x_0}(B_r)$. Thus we can fix smooth vector fields $L_1,L_2\in\mathbb{T}(B_{\eps_0}(x_0))$ such that
\begin{equation}\label{gt2.2}
\begin{split}
&\g(L_1,L_1)=\g(L_2,L_2)=1;\\
&\g(L_1,L_2)=\g(L_1,L_3)=\g(L_2,L_3)=\g(L_1,L_4)=\g(L_2,L_4)=0.
\end{split}
\end{equation}
We define also the smooth function $N^{x_0}:B_1(x_0)\to[0,\infty)$
\begin{equation*}
N^{x_0}(x)=|(\Phi^{x_0})^{-1}(x)|^2.
\end{equation*} 
The main result in this section is the following Carleman estimate:

\begin{proposition}\label{Car}
There is $\eps\in (0,\ep_0)$ sufficiently small and $C_\eps$ sufficiently large such that for any $\lambda\geq C_\eps$ and any $\phi\in C^\infty_0(B_{\eps^{10}}(x_0))$
\begin{equation}\label{Car1}
\lambda \|e^{-\lambda f_\eps}\phi\|_{L^2}+\|e^{-\lambda f_\eps}|D^1\phi|\,\|_{L^2}\leq C_\eps\lambda^{-1/2}\|e^{-\lambda f_\eps}\,\square_{\g}\phi\|_{L^2},
\end{equation}
where
\begin{equation}\label{Car2}
f_\eps=\ln[\ep^{-1}(u+\eps)(v+\eps)+\eps^{12}N^{x_0}].
\end{equation}
\end{proposition}

\begin{proof}[Proof of Proposition \ref{Car}] We apply Proposition \ref{Cargen} with $V=0$. It is clear that $\eps^{12}N^{x_0}$ is a negligible perturbation, in the sense of \eqref{smallweight}, for $\eps$ sufficiently small. It remains to prove that there is $\eps_1>0$ such that the family of weights $\{h_\eps\}_{\eps\in (0,\eps_1)}$,
\begin{equation}\label{Car2.2}
h_\eps=\eps^{-1}(u+\eps)(v+\eps)
\end{equation}
satisfies conditions \eqref{po5}, \eqref{po3.2} and \eqref{po3}.

Let $\widetilde{C}$ denote constants $\geq 1$ that may depend only on the predefined geometric quantities $\eps_0$, $\delta_0$, and a uniform bound in $B_{\eps_0}(x_0)$ of $|D^j\g_{\al\be}|$, $|D^j\g^{\al\be}|$, $|D^j u|$, $|D^j v|$, $j=0,\ldots,6$. Since $u(x_0)=v(x_0)=0$, the definition \eqref{Car2.2} shows easily that condition \eqref{po5} is satisfied, provided that $\eps_1\leq \widetilde{C}^{-1}$.

Relative to the frame $L_1, L_2, L_3,
L_4$ the metric $\g$ takes the form,
\begin{equation}\label{Car30}
\begin{cases}
&\g_{ab}=\de_{ab}, \quad \g_{a3}=\g_{a4}=0,\quad   a,b=1,2\\
&\g_{33}=g_3,\quad \g_{44}=g_4,  \quad \g_{34}= \Omega,
\end{cases}
\end{equation}
in $B_{\eps_0}(x_0)$, where $g_3=\g(L_3,L_3)$, $g_4=\g(L_4,L_4)$, $\Omega=\g(L_3,L_4)$. Also, for the inverse metric,
\begin{equation}\label{Car52}
\begin{cases}
&\g^{ab}=\de^{ab}, \quad \g^{a3}=\g^{a4}=0,\quad   a,b=1,2\\
&  \g^{33}=g'_3,\quad \g^{44}=g'_4,  \quad \g^{34}= \Omega',
\end{cases}
\end{equation}
where $g'_3=-g_4/(\Omega^2-g_3g_4)$, $g'_4=-g_3/(\Omega^2-g_3g_4)$, $\Omega'=\Omega/(\Omega^2-g_3g_4)$. Recall that $\Omega\geq \delta_0/2$ in $B_{\eps_0}(x_0)$, see \eqref{gt2}, $g_3=0$ on $\mathcal{H}^+\cap B_{\eps_0}(x_0)$, $g_4=0$ on $\mathcal{H}^-\cap B_{\eps_0}(x_0)$, see \eqref{ao1}. Thus, using Lemma \ref{ao} (c),
\begin{equation}\label{ao99}
|g_3|\leq \widetilde{C}u\quad\text{and}\quad|g_4|\leq\widetilde{C}v\quad\text{ in }B_{\eps_0}(x_0).
\end{equation}

We denote by $O(1)$ any quantity with absolute value bounded by a constant $\widetilde{C}$ as before. In view of the definitions of $u,v,L_1,L_2,L_3,L_4$ 
we have,
\begin{equation}\label{lubounds}
L_1(u)=L_2(u)=L_1(v)=L_2(v)=0,\quad L_3(u)=g_3,\quad L_4(v)=g_4,\quad L_4(u)=L_3(v)=\Om.
\end{equation}
Thus
\begin{equation}\label{eq:Car1.1}
\begin{split}
&L_4(h_\ep)=\eps^{-1}(v+\ep)\Om+\eps^{-1}(u+\eps)g_4,
\quad L_3(h_\ep)=\eps^{-1}(u+\eps)\Om+\eps^{-1}(v+\eps)g_3,\\
&L_1(h_\ep)=L_2(h_\ep)=0,
\end{split}
\end{equation}
and, using \eqref{ao99}, \eqref{lubounds}, and \eqref{eq:Car1.1}, in $B_{\ep^{10}}(x_0)$,
\bea\label{car:D2f}
\begin{cases}
&(\D^2h_\ep)_{34}=(\D^2h_\ep)_{43}=\eps^{-1}\Om^2+O(1),\\
&(\D^2h_\ep)_{33}=O(1),\quad(\D^2h_\ep)_{44}=O(1),\quad(\D^2h_\ep)_{ab}=O(1),\qquad a,b=1,2,\\
&(\D^2h_\ep)_{3a}=O(1),\quad(\D^2h_\ep)_{4a}=O(1),\qquad a=1,2.
\end{cases}
\eea

Using \eqref{Car52}, \eqref{eq:Car1.1}, \eqref{car:D2f}, and $g_3(x_0)=g_4(x_0)=0$ we compute
\begin{equation*}
\D^\al h_\eps(x_0)\D^\be h_\eps(x_0)(\D_\al h_\eps\D_\be h_\eps-\eps\D_\al\D_\be h_\eps)(x_0)=2\Omega^2+\eps O(1)\geq \delta_0^2
\end{equation*}
if $\eps_1$ is sufficiently small. Thus condition \eqref{po3.2} is satisfied provided $\eps_1\leq \widetilde{C}^{-1}$.

Assume now $Y=Y^\al L_\al$ is a vector in $\mathbb{T}_{x_0}(\mathbf{M}^4)$. We fix $\mu=\eps_1^{-1/2}$ and compute, using \eqref{eq:Car1.1}, \eqref{car:D2f}, and $g_3(x_0)=g_4(x_0)=1$,
\begin{equation*}
\begin{split}
&Y^\al Y^\be(\mu\g_{\al\be}-\D_\al\D_\be h_\eps)(x_0)+\eps^{-2}|Y^\al\D_\al h_\eps|^2\\
&=\mu((Y^1)^2+(Y^2)^2+2\Omega Y^3Y^4)-2\eps^{-1}\Omega^2Y^3Y^4+\eps^{-2}\Omega^2(Y^3+Y^4)^2+O(1)\sum_{\al=1}^4(Y^\al)^2\\
&\geq (\mu/2)[(Y^1)^2+(Y^2)^2]+\Omega^2(\eps^{-1}/2)[(Y^3)^2+(Y^4)^2]\\
&\geq (Y^1)^2+(Y^2)^2+(Y^3)^2+(Y^4)^2
\end{split}
\end{equation*}
if $\eps_1$ is sufficiently small. We notice now that we can write $Y=X^\al\partial_\al$ in the coordinate frame $\partial_1,\partial_2,\partial_3,\partial_4$, and $|X^\alpha|\leq \widetilde{C}(|Y^1|+|Y^2|+|Y^3|+|Y^4|)$ for $\al=1,2,3,4$. Thus condition \eqref{po3} is satisfied provided $\eps_1\leq \widetilde{C}^{-1}$, which completes the proof of the lemma.
\end{proof}

\subsection{The second Carleman inequality in Kerr spaces}\label{CarEst2}

In this section we prove a Carleman estimate for functions supported in small open sets in $\mathbf{E}^4$. Assume that $x_0\in\mathbf{E}^4$ and $\Phi^{x_0}:B_1\to\mathbf{E}^4$, $\Phi^{x_0}(0)=x_0$, is a smooth coordinate chart around $x_0$. We define the smooth function $N^{x_0}:B_1(x_0)\to[0,\infty)$, $N^{x_0}(x)=|(\Phi^{x_0})^{-1}(x)|^2$ as before.

We use the notation in the appendix. The coordinate function $r:\widetilde{\mathbf{E}}^4\to(r_+,\infty)$ extends to a smooth function $r:\mathbf{E}^4\to(r_+,\infty)$. The main result in this subsection is the following Carleman estimate:

\begin{proposition}\label{Carl2}
There is $\eps\in (0,1/2]$ sufficiently small and $\widetilde{C}_\eps$ sufficiently large such that for any $\lambda\geq\widetilde{C}_\eps$ and any $\phi\in C^\infty_0(B_{\eps^{10}}(x_0))$
\begin{equation}\label{Carb1}
\lambda \|e^{-\lambda \f_\eps}\phi\|_{L^2}+\|e^{-\lambda \f_\eps}|D^1\phi\,|\,\|_{L^2}\leq \widetilde{C}_\eps\lambda^{-1/2}\|e^{-\lambda \f_\eps}\,\square_{\g}\phi\|_{L^2}+\eps^{-6}\|e^{-\lambda \f_\eps}\xi(\phi)\,\|_{L^2},
\end{equation}
where, with $r_0=r(x_0)$,
\begin{equation}\label{Carb2}
\f_\eps=\ln[r-r_0+\eps+\eps^{12}N^{x_0}].
\end{equation}
\end{proposition}

\begin{proof}[Proof of Proposition \ref{Carl2}] As in the proof of Proposition \ref{Car}, we will use the notation $\widetilde{C}$ to denote various constants in $[1,\infty)$ that may depend only on the chart $\Phi$ and the position of $x_0$  in $\mathbf{E}^4$ (i.e. on $(r(x_0)-r_+)^{-1}+(r(x_0)-r_+)]$), and $O(1)$ to denote quantities bounded in absolute value by a constant $\widetilde{C}$. It is important to keep in mind that $r(x_0)>r_+$, i.e. $x_0\in\mathbf{E}^4$. We apply Proposition \ref{Cargen} with $V=\xi$. It suffices to prove that there is $\eps_1>0$ such that the family of weights $h_\eps\}_{\eps\in(0,\eps_1)}$,
\begin{equation}\label{weight2}
h_\eps=r-r_0+\eps
\end{equation}
satisfies conditions \eqref{po5}, \eqref{po3.2}, and \eqref{po3}.

Condition \eqref{po5} is clear if $\eps_1$ is sufficiently small, since $\xi(h_\eps)=0$. To prove conditions \eqref{po3.2} and \eqref{po3}, with the notation in section \ref{explicit}, we work in the orthonormal frame $e_0,e_1,e_2,e_3$ defined in \eqref{g3}. We have
\begin{equation}\label{tp1}
\D_0(h_\eps)=\D_1(h_\eps)=\D_3(h_\eps)=0,\qquad \D_2(h_\eps)=(\Delta/\rho^2)^{1/2}.
\end{equation}
Using the table \eqref{de1}, we have
\begin{equation}\label{tp2}
\begin{split}
-&\D_0\D_0 h_\ep=\frac{\Delta}{\rho^2}\Big(\frac{r}{\rho^2}+\frac{r-m}{\Delta}-\frac{Y}{\Sigma^2}\Big)\\
-&\D_0\D_1 h_\ep=-\frac{\Delta}{\rho^2}\cdot \frac{ma\sin\theta}{\rho^2\sqrt{\Delta}\Sigma^2}(2rY-\Sigma^2)\\
-&\D_1\D_1 h_\ep=-\frac{\Delta}{\rho^2}\Big(\frac{Y}{\Sigma^2}-\frac{r}{\rho^2}\Big)\\
-&\D_2\D_2 h_\ep=\frac{\Delta}{\rho^2}\Big(\frac{r}{\rho^2}-\frac{r-m}{\Delta}\Big)\\
-&\D_2\D_3 h_\ep=-\frac{\sqrt{\Delta}a^2\sin\theta\cos\theta}{\rho^4}\\
-&\D_3\D_3 h_\ep=-\frac{\Delta r}{\rho^4}\\
 &\D_0\D_2 h_\ep=\D_0\D_3 h_\ep=\D_1\D_2 h_\ep=\D_1\D_3 h_\ep=0.
\end{split}
\end{equation}
It follows that
\begin{equation*}
\D^\alpha h_\eps(x_0)\D^\be h_\eps(x_0)(\D_\al h_\eps\D_\be h_\eps-\eps\D_\al\D_\be h_\eps)(x_0)=\Delta^2/\rho^2+\eps O(1)
\end{equation*}
which verifies condition \eqref{po3.2} if $\eps_1$ is sufficiently small.

To verify condition \eqref{po3} we fix
\begin{equation}\label{tp3}
\mu=\frac{3\Delta r}{2\rho^4} 
\end{equation}
and use the formula (compare with \eqref{g10} and \eqref{k3})
\begin{equation}\label{Carb65}
\xi=\frac{\rho\sqrt{\Delta}}{\Sigma}e_0-\frac{2amr\sin\theta}{\rho\Sigma}e_1.
\end{equation}
Assume $X=Y^0e_0+Y^1e_1+Y^2e_2+Y^3e_3$ is a vector expressed in the frame $e_\al$. We compute
\begin{equation}\label{tp5}
\begin{split}
&Y^\al Y^\be(\mu\g_{\al\be}-\D_\al\D_\be h_\eps)(x_0)+\eps^{-2}(|Y^\al \xi_\al(x_0)|^2+|Y^\al\D_\al h_\eps(x_0)|^2)\\
&=(Y^0)^2(-\mu-\D_0\D_0 h_\eps)+(Y^1)^2(\mu-\D_1\D_1 h_\eps)+2Y^0Y^1(-\D_0\D_1h_\eps)\\
&+(Y^2)^2(\mu-\D_2\D_2 h_\eps)+(Y^3)^2(\mu-\D_3\D_3 h_\eps)+2Y^2Y^3(-\D_2\D_3h_\eps)\\
&+\eps^{-2}\frac{(\rho^2\sqrt{\Delta}Y^0+2amr(\sin\theta)Y^1)^2}{\rho^2\Sigma^2}+\eps^{-2}\frac{\Delta(Y^2)^2}{\rho^4}.
\end{split}
\end{equation}
Let $Z=\rho^2\sqrt{\Delta}Y^0+2amr(\sin\theta)Y^1$, thus 
\begin{equation*}
Y^0=\frac{Z-2amr(\sin\theta)Y^1}{\rho^2\sqrt{\Delta}}=\alpha Y^1+\beta Z.
\end{equation*}
Using also $\mu-\D_3\D_3h_\ep=(\Delta r)/(2\rho^4)$, the right-hand side of \eqref{tp5} becomes
\begin{equation}\label{tp6}
\begin{split}
&(Y^2)^2(\eps^{-2}\Delta\rho^{-4}+\mu-\D_2\D_2 h_\eps)+(Y^3)^2(\Delta r\rho^{-4})/2-2Y^2Y^3\cdot\D_2\D_3h_\eps\\
&+Z^2[\eps^{-2}\rho^{-2}\Sigma^{-2}+\beta^2(-\mu-\D_0\D_0h_\eps)]\\
&+(Y^1)^2[\al^2(-\mu-\D_0\D_0h_\eps)-2\alpha\D_0\D_1h_\eps+\mu-\D_1\D_1h_\eps]\\
&+2Y^1Z[\al\be(-\mu-\D_0\D_0h_\eps)-\be\D_0\D_1h_\eps].
\end{split}
\end{equation}
It is clear that the first line of the expression above is bounded from below by $$\widetilde{C}^{-1}(\eps^{-2}(Y^2)^2+(Y^3)^2)$$
if $\eps$ is sufficiently small, since $\Delta\geq\widetilde{C}^{-1}$. The main term we need to bound from below is the coefficient of $(Y^1)^2$ in \eqref{tp6}. We use the table \eqref{tp2} and the definitions of $\alpha$ and $\mu$; after several simplifications this term is equal to
\begin{equation*}
\frac{5\Delta r}{2\rho^4}-\frac{\Delta Y}{\rho^2\Sigma^2}+\frac{4a^2m^2r(\sin\theta)^2}{\rho^6}\Big(-\frac{r^2}{2\rho^2}+\frac{rY}{\Sigma^2}+\frac{mr-a^2}{\Delta}\Big).
\end{equation*}
In view of \eqref{de5} and \eqref{de6} this is bounded from below by $(\Delta r)/(2\rho^4)$. Thus the sum of the last three lines of \eqref{tp6} is bounded from below by $$\widetilde{C}^{-1}(\eps^{-2}Z^2+(Y^1)^2)$$
if $\eps$ is sufficiently small. It follows that
\begin{equation*}
\begin{split}
&Y^\al Y^\be(\mu\g_{\al\be}-\D_\al\D_\be h_\eps)(x_0)+\eps^{-2}(|Y^\al \xi_\al(x_0)|^2+|Y^\al\D_\al h_\eps(x_0)|^2)\\
&\geq\widetilde{C}^{-1}[(Y^0)^2+(Y^1)^2+\eps^{-2}(Y^2)^2+(Y^3)^2]
\end{split}
\end{equation*}
if $\eps$ is sufficiently small. The condition \eqref{po3} is verified, which completes the proof of the proposition.
\end{proof}

\subsection{Vanishing of the tensor $\mathcal{S}$}\label{proof}

In this subsection we prove Theorem \ref{Main1kerr}. Arguments showing how to use Carleman inequalities to prove uniqueness are standard. We provide all the details here for the sake of completeness. Some care is needed at the first step, in Lemma \ref{Lemmaa1} below, since we do not assume that derivatives of the tensor $\Ss$ vanish on the horizon.

We show first that the tensor $\Ss$ vanishes in a neighborhood of the bifurcate sphere $\, S_0$ in $\mathbf{E}^4$.

\begin{lemma}\label{Lemmaa1}
With the notation in Theorem \ref{Main1kerr}, there is an open set $\mathbf{O}'\subseteq{\mathbf{M}}^4$, $\, S_0\subseteq \mathbf{O}'$, such that
\begin{equation*}
\Ss\equiv 0\text{ in }\mathbf{O}'\cap\mathbf{E}^4.
\end{equation*}
\end{lemma}

\begin{proof}[Proof of Lemma \ref{Lemmaa1}] We use the functions $u,v$ defined in Lemma \ref{ao} and the Carleman estimate in Proposition \ref{Car}. Since $\, S_0$ is compact, it suffices to prove that for every point $x_0\in\, S_0$ there is a neighborhood $\mathbf{O}'_{x_0}$ of $x_0$ such that $\Ss\equiv 0$ in $\mathbf{E}^4\cap\mathbf{O}'_{x_0}$. As in Proposition \ref{Car}, assume $\Phi^{x_0}:B_1\to\mathbf{O}$, $\Phi^{x_0}(0)=x_0$, is a smooth coordinate chart around $x_0$. With the notation in Proposition \ref{Car}, there are constants $\eps\in(0,\eps_0)$ and $\widetilde{C}\geq 1$ such that, for any $\lambda\geq \widetilde{C}$ and any $\phi\in C^\infty_0(B_{\eps^{10}}(x_0))$
\begin{equation}\label{Bar1}
\lambda \|e^{-\lambda f_\eps}\phi\|_{L^2}+\|e^{-\lambda f_\eps}|D^1\phi|\,\|_{L^2}\leq \widetilde{C}\lambda^{-1/2}\|e^{-\lambda f_\eps}\,\square_{\g}\phi\|_{L^2},
\end{equation}
where
\begin{equation}\label{Bar2}
f_\eps=\ln[\eps^{-1}(u+\eps)(v+\eps)+\eps^{12}N^{x_0}].
\end{equation}
The constant $\eps$ will remain fixed in this proof, and we assume implicitly it is sufficiently small as discussed in Proposition \ref{Car}. We will show that
\begin{equation}\label{Bar3}
\Ss\equiv 0\text{ in }B_{\eps^{40}}(x_0)\cap\mathbf{E}^4.
\end{equation}

For $(j_1,\ldots,j_k)\in\{1,2,3,4\}^k$ we define, using the coordinate chart $\Phi$,
\begin{equation}\label{va5}
\phi_{(j_1\ldots j_k)}=\mathcal{S}(\partial_{j_1},\ldots ,\partial_{j_k}).
\end{equation}
If $k=0$ we define $\phi=\Ss$ in $B_1(x_0)$. The functions $\phi_{(j_1\ldots j_k)}:B_1(x_0)\to\mathbb{C}$ are smooth. Let $\eta:\mathbb{R}\to[0,1]$ denote a smooth function supported
in $[1/2,\infty)$ and equal to $1$ in $[3/4,\infty)$. With $u,v$ as in Proposition \ref{ao}, for $\delta \in(0,1]$ we define
\begin{equation}\label{pr2}
\begin{split}
\phi^{\delta,\eps}_{(j_1\ldots j_k)}&=\phi_{(j_1\ldots j_k)}\cdot \mathbf{1}_{\mathbf{E}^4}\cdot \eta(uv/ \delta)\cdot \big(1-\eta(N^{x_0}/ \eps^{20})\big)\\
&=\phi_{(j_1\ldots j_k)}\cdot \widetilde{\eta}_{\delta,\eps}.
\end{split}
\end{equation}
Clearly, $\phi^{\delta,\eps}_{(j_1\ldots j_k)}\in C^\infty _0(B_{\eps^{10}}(x_0))$. We would like to apply the inequality \eqref{Bar1} to the functions $\phi^{\delta,\eps}_{(j_1\ldots j_k)}$, and then let $\delta \to 0$ and $\lambda \to\infty$ (in this order).

Using the definition \eqref{pr2}, we have
\begin{equation*}
\square_\g\phi^{\delta,\eps}_{(j_1\ldots j_k)}=\widetilde{\eta}_{\delta,\eps}\cdot\square_\g\phi_{(j_1\ldots j_k)}+2\D_\al\phi_{(j_1\ldots j_k)}\cdot \D^\al \widetilde{\eta}_{\delta,\eps}+\phi_{(j_1\ldots j_k)}\cdot \square_\g\widetilde{\eta}_{\delta,\eps}.
\end{equation*}
Using the Carleman inequality \eqref{Bar1}, for any $(j_1,\ldots j_k)\in\{1,2,3,4\}^k$ we have
\begin{equation}\label{va10}
\begin{split}
&\lambda\cdot \|e^{-\lambda f_{\eps}}\cdot \widetilde{\eta}_{\delta,\eps}\phi_{(j_1\ldots j_k)} \|_{L^2}+\|e^{-\lambda f_{\eps}}\cdot \widetilde{\eta}_{\delta,\eps}|D^1\phi_{(j_1\ldots j_k)}|\, \|_{L^2}\\
&\leq \widetilde{C}\lambda ^{-1/2}\cdot \|e^{-\lambda f_{\eps}}\cdot \widetilde{\eta}_{\delta,\eps}\square_\g\phi_{(j_1\ldots j_k)}\|_{L^2}\\
&+\widetilde{C}\Big[\|e^{-\lambda f_{\eps}}\cdot \D_\al\phi_{(j_1\ldots j_k)}\D^\al \widetilde{\eta}_{\delta,\eps} \|_{L^2}+\|e^{-\lambda f_{\eps}}\cdot \phi_{(j_1\ldots j_k)} ( |\square_\g\widetilde{\eta}_{\delta,\eps}|+|D^1\widetilde{\eta}_{\delta,\eps}| )\|_{L^2}\Big],
\end{split}
\end{equation}
for any $\lambda\geq \widetilde{C}$. We estimate now $|\square_\g\phi_{(j_1\ldots j_k)}|$. Using the first identity in \eqref{va4} and \eqref{va5}, in $B_{\eps^{10}}(x_0)$ we estimate pointwise
\begin{equation}\label{va11}
|\square_\g\phi_{(j_1\ldots j_k)}|\leq \widetilde{C}_{\Aa,\Bb}\sum_{l_1,\ldots,l_k}\big(|\,D^1\phi_{(l_1\ldots l_k)}|+|\phi_{(l_1\ldots l_k)}|\big),
\end{equation}
for some constant $\widetilde{C}_{\Aa,\Bb}$ that depends only on the tensors $\Aa$ and $\Bb$. We add up the inequalities \eqref{va10} over $(j_1,\ldots,j_k)\in\{1,2,3,4\}^k$. The key observation is that, in view of \eqref{va11}, the first term in the right-hand side can be absorbed into the left-hand side for $\lambda$ sufficiently large. Thus, for any $\lambda\geq \widetilde{C}_{\Aa,\Bb}$ and $\delta \in(0,1]$
\begin{equation}\label{va12}
\begin{split}
&\lambda\sum_{j_1,\ldots ,j_k}\|e^{-\lambda f_{\eps}}\cdot \widetilde{\eta}_{\delta,\eps}\phi_{(j_1\ldots j_k)} \|_{L^2}\\
&\leq \widetilde{C}\sum_{j_1,\ldots ,j_k}\Big[\|e^{-\lambda f_{\eps}}\cdot \D_\al\phi_{(j_1\ldots j_k)}\D^\al \widetilde{\eta}_{\delta,\eps} \|_{L^2}+\|e^{-\lambda f_{\eps}}\cdot \phi_{(j_1\ldots j_k)} ( |\square_\g\widetilde{\eta}_{\delta,\eps}|+|D^1\widetilde{\eta}_{\delta,\eps}| )\|_{L^2}\Big].
\end{split}
\end{equation}

We would like to let $\delta\to 0$ in \eqref{va12}. For this, we observe first that the functions $\D_\al\phi_{(j_1\ldots j_k)}\D^\al \widetilde{\eta}_{\delta,\eps}$ and $( |\square_\g\widetilde{\eta}_{\delta,\eps}|+|D^1\widetilde{\eta}_{\delta,\eps}| )$ vanish outside the set $\mathbf{A}_\delta\cup\widetilde{\mathbf{B}}_\eps$, where
\begin{equation*}
\begin{cases}
&\mathbf{A}_\delta=\{x\in B_{\eps^{10}}(x_0)\cap\mathbf{E}^4:uv\in(\delta/2,\delta )\};\\
&\widetilde{\mathbf{B}}_\eps=\{x\in B_{\eps^{10}}(x_0)\cap\mathbf{E}^4:N^{x_0}\in(\eps^{20}/2,\eps^{20})\}.
\end{cases}
\end{equation*}
In addition, since $\phi_{(j_1\ldots j_k)}=0$ on $\HH$ (using the hypothesis of Theorem \ref{Main1}), it follows from Proposition \ref{ao} (c) that there are smooth functions ${\phi'}_{(j_1\ldots j_k)}:\mathbf{O}\to\mathbb{C}$ such that
\begin{equation}\label{va20}
\phi_{(j_1\ldots j_k)}\big(1-\eta(N^{x_0})\big)=uv\cdot {\phi'}_{(j_1\ldots j_k)}\text{ in }\mathbf{O}\cap\mathbf{E}^4.
\end{equation}

We show now that
\begin{equation}\label{va23}
\begin{split}
|\square_\g\widetilde{\eta}_{\delta,\eps}|+|D^1\widetilde{\eta}_{\delta,\eps}|&\leq \widetilde{C}(\mathbf{1}_{\widetilde{\mathbf{B}}_\eps}+(1/ \delta)\mathbf{1}_{\mathbf{A}_\delta}).
\end{split}
\end{equation}
The  inequality for $|D^1\widetilde{\eta}_{\delta,\eps}|$ follows directly from the definition \eqref{pr2}. Also, using again the definition,
\begin{equation*}
|\D^\al\D_\al\widetilde{\eta}_{\delta,\eps}|\leq |\D^\al\D_\al(\mathbf{1}_{\mathbf{E}^4}\cdot \eta(uv/ \delta))|\cdot \big(1-\eta(N^{x_0}/ \eps^{20})\big)+\widetilde{C}(\mathbf{1}_{\widetilde{\mathbf{B}}_\eps}+(1/ \delta)\mathbf{1}_{\mathbf{A}_\delta}).
\end{equation*}
Thus, for \eqref{va23}, it suffices to prove that
\begin{equation}\label{va23.7}
\mathbf{1}_{\mathbf{E}^4}\cdot |\D^\al\D_\al(\eta(uv/ \delta))|\leq \widetilde{C}/ \delta\cdot \mathbf{1}_{\mathbf{A}_\delta}.
\end{equation}
Since $u,v,\eta$ are smooth functions, for \eqref{va23.7} it suffices to prove that
\begin{equation}\label{va89}
\delta^{-2}|\D^\al(uv)\D_\al(uv)|\leq\widetilde{C}/ \delta\text{ in }\mathbf{A}_\delta. 
\end{equation}
Since $uv\in [\delta/2,\delta]$ in $\mathbf{A}_\delta$, it suffices to prove that
\begin{equation*}
u^2|\D^\al v\D_\al v|+v^2|\D^\al u\D_\al u|\leq \widetilde{C}\delta\text{ in }\mathbf{A}_\delta.
\end{equation*}
For this we use the frame $L_1,L_2,L_3,L_4$ as  in the proof of Prooposition \ref{Car}. The bound follows from \eqref{lubounds}, \eqref{ao99}, and \eqref{Car52}.
 
We show now that 
\begin{equation}\label{va21}
|\D_\al\phi_{(j_1\ldots j_k)}\D^\al \widetilde{\eta}_{\delta,\eps}|\leq  \widetilde{C}_{\phi'}(\mathbf{1}_{\widetilde{\mathbf{B}}_\eps}+\mathbf{1}_{\mathbf{A}_\delta}),
\end{equation}
where the constant $\widetilde{C}_{\phi'}$ depends on the smooth functions ${\phi'}_{(j_1\ldots j_k)}$ defined in \eqref{va20}. Using the formula \eqref{va20} (which becomes $\phi_{(j_1\ldots j_k)}=uv\cdot {\phi'}_{(j_1\ldots j_k)}$ in $\mathbf{A}_\delta\cup\widetilde{\mathbf{B}}_\eps$), this follows easily from \eqref{va89}.

It follows from \eqref{va20}, \eqref{va23}, and \eqref{va21} that
\begin{equation*}
|\D_\al\phi_{(j_1\ldots j_k)}\D^\al \widetilde{\eta}_{\delta,\eps}|+|\phi_{j_1\ldots j_k}|( |\square_\g\widetilde{\eta}_{\delta,\eps}|+|D^1\widetilde{\eta}_{\delta,\eps}| )\leq\widetilde{C}_{\phi'}(\mathbf{1}_{\widetilde{\mathbf{B}}_\eps}+\mathbf{1}_{\mathbf{A}_\delta}). 
\end{equation*}
Since $\lim_{\delta\to 0}\|\mathbf{1}_{\mathbf{A}_\delta}\|_{L^2}=0$, we can let $\delta\to
0$ in \eqref{va12} to conclude that
\begin{equation}\label{va13}
\lambda\sum_{j_1,\ldots ,j_k}\|e^{-\lambda f_{\eps}}\cdot \mathbf{1}_{B_{\eps^{10}/2}(x_0)\cap\mathbf{E}^4}\cdot \phi_{(j_1\ldots j_k)} \|_{L^2}\leq \widetilde{C}_{\phi'}\sum_{j_1,\ldots ,j_k}\|e^{-\lambda f_{\eps}}\cdot \mathbf{1}_{\widetilde{\mathbf{B}}_\eps}\|_{L^2}
\end{equation}
for any $\lambda\geq \widetilde{C}_{\Aa,\Bb}$. Finally, using the definition \eqref{Bar2}, we observe that
\begin{equation*}
\inf_{B_{\ep^{40}}(x_0)\cap\mathbf{E}^4}e^{-\lambda f_{\eps}}\geq e^{-\lambda \ln(\eps+\eps^{32}/2)}\geq \sup_{\widetilde{\mathbf{B}}_\eps}e^{-\lambda f_{\eps}}.
\end{equation*}
It follows from \eqref{va13} that
\begin{equation*}
\lambda\sum_{j_1,\ldots ,j_k}\|\mathbf{1}_{B_{\ep^{40}}(x_0)\cap\mathbf{E}^4}\cdot \phi_{(j_1\ldots j_k)} \|_{L^2}\leq \widetilde{C}_{\phi'}\sum_{j_1,\ldots ,j_k}\|\mathbf{1}_{\widetilde{\mathbf{B}}_\eps}\|_{L^2}
\end{equation*}
for any $\lambda\geq \widetilde{C}_{\Aa,\Bb}$. We let now $\lambda\to\infty$. The identity \eqref{Bar3} follows.
\end{proof}

We show now that the tensor $\Ss$ vanishes in an open neighborhood of the horizon $\HH$ in $\mathbf{E}^4$. For any $R>r_+$ let
\begin{equation*}
\mathbf{E}^4_R=\{x\in\mathbf{E}^4:r(x)\in(r_+,R)\},
\end{equation*}
where $r:\mathbf{E}^4\to(r_+,\infty)$ is the smooth function used in Proposition \ref{Carl2}.

\begin{lemma}\label{Lemmaa2}
With the notation in Theorem \ref{Main1kerr}, there is $R>r_+$ such that
\begin{equation*}
\Ss\equiv 0\text{ in }\mathbf{E}^4_R.
\end{equation*}
\end{lemma}

\begin{proof}[Proof of Lemma \ref{Lemmaa2}] It follows from Proposition \ref{ao} (a) and Lemma \ref{Lemmaa1} that there is $\eps_1>0$ such that
\begin{equation}\label{tr1}
\Ss\equiv 0\text{ in the set }\{x\in\mathbf{E}^4\cap\mathbf{O}: u(x)<\eps_1\text{ and }v(x)<\eps_1\}.
\end{equation}
It suffices to prove that $\Ss\equiv 0$ in $\mathbf{E}^4_R\cap\widetilde{\mathbf{E}}^4$, where $\widetilde{\mathbf{E}}^4$ is the dense open subset of $\mathbf{E}^4$ defined in section \ref{explicit}. In view of \eqref{tr1} and the definition of the functions $u,v$ in the proof of Lemma \ref{ao}, there is $\eps_2>0$ such that
\begin{equation}\label{tr2}
\Ss\equiv 0\text{ in the set }\{x=(r,t,\theta,\phi)\in\widetilde{\mathbf{E}}^4: t=0\text{ and }r<r_++\eps_2\}.
\end{equation}

We use the Boyer-Lindquist coordinate chart (see appendix \ref{explicit}) to define
\begin{equation*}
\widetilde{\partial}_1=\partial_r,\,\,\widetilde{\partial}_2=\partial_t,\,\,\widetilde{\partial}_3=\partial_\theta,\,\,\widetilde{\partial}_4=\partial_\phi
\end{equation*}
and
\begin{equation*}
\widetilde{\phi}_{(j_1\ldots j_k)}=\Ss(\widetilde{\partial}_{j_1},\ldots,\widetilde{\partial}_{j_k})
\end{equation*}
The second identity in \eqref{va4} gives, for any $(j_1,\ldots,j_k)\in\{1,2,3,4\}^k$
\begin{equation}\label{tr90}
\partial_t(\widetilde{\phi}_{(j_1\ldots j_k)})=\sum_{l_1,\ldots,l_k}\widetilde{\phi}_{(l_1\ldots l_k)}{\Cc^{l_1\ldots l_k}}_{j_1\ldots j_k}.
\end{equation}
In view of \eqref{tr2}
\begin{equation*}
\widetilde{\phi}_{(j_1\ldots j_k)}(r,0,\theta,\phi)=0\text{ if }r<r_++\eps_2.
\end{equation*}
Since $\Cc$ is a smooth tensor in $\widetilde{\mathbf{E}}^4$, it follows that $\widetilde{\phi}_{(j_1\ldots j_k)}(r,t,\theta,\phi)=0$ if $r<r_++\eps_2$, which completes the proof of the lemma. 
\end{proof}

We prove now that $\Ss\equiv 0$ in $\mathbf{E}^4$, which completes the proof of the theorem. In view of Lemma \ref{Lemmaa2}, it suffices to prove the following:

\begin{lemma}\label{Lemmaa3}
With the notation in Theorem \ref{Main1}, assume that
\begin{equation}\label{tr70}
\Ss\equiv 0\text{ in }\mathbf{E}^4_{R_0}.
\end{equation}
for some $R_0>r_+$. Then there is $R_1>R_0$ such that
\begin{equation*}
\Ss\equiv 0\text{ in }\mathbf{E}^4_{R_1}.
\end{equation*}
\end{lemma}

\begin{proof}[Proof of Lemma \ref{Lemmaa3}] Assume that $x_0\in\mathbf{E}^4$ and $r(x_0)=R_0$. We show first that 
\begin{equation}\label{tr10}
\text{ there is a neighborhood }\mathbf{O}'_{x_0}\text{ of }x_0\text{ such that }\Ss\equiv 0\text{ in }\mathbf{O}'_{x_0}.
\end{equation}
This is similar to the proof of Lemma \ref{Lemmaa1}, using the Carleman estimate in Proposition \ref{Carl2} instead of the Carleman estimate in Proposition \ref{Car}. Assume $\Phi^{x_0}:B_1\to\mathbf{E}^4$, $\Phi^{x_0}(0)=x_0$, is a smooth coordinate chart around $x_0$. With the notation in Proposition \ref{Carl2}, there is $\eps\in (0,1/2]$ sufficiently small and $\widetilde{C}$ sufficiently large such that 
\begin{equation}\label{tr11}
\lambda \|e^{-\lambda \f_\eps}\phi\|_{L^2}+\|e^{-\lambda \f_\eps}|D^1\phi\,|\,\|_{L^2}\leq \widetilde{C}\lambda^{-1/2}\|e^{-\lambda \f_\eps}\,\square_{\g}\phi\|_{L^2}+\eps^{-6}\|e^{-\lambda \f_\eps}\xi(\phi)\,\|_{L^2},
\end{equation}
for any $\lambda\geq\widetilde{C}$ and any $\phi\in C^\infty_0(B_{\eps^{10}}(x_0))$, where
\begin{equation}\label{tr12}
\f_\eps=\ln[r-R_0+\eps+\eps^{12}N^{x_0}].
\end{equation}
The constant $\eps$ will remain fixed in this proof, and sufficiently small in the sense of Proposition \ref{Carl2}. We will show that
\begin{equation}\label{Tar3}
\Ss\equiv 0\text{ in }B_{\eps^{40}}(x_0).
\end{equation}

For $(j_1,\ldots,j_k)\in\{1,2,3,4\}^k$ we define, using the coordinate chart $\Phi$,
\begin{equation*}
\phi_{(j_1\ldots j_k)}=\mathcal{S}(\partial_{j_1},\ldots ,\partial_{j_k}).
\end{equation*}
If $k=0$ we simply define $\phi=\Ss$ in $B_1(x_0)$. The functions $\phi_{(j_1\ldots j_k)}:B_1(x_0)\to\mathbb{C}$ are smooth. Let $\eta:\mathbb{R}\to[0,1]$ denote a smooth function supported
in $[1/2,\infty)$ and equal to $1$ in $[3/4,\infty)$, as before. We define
\begin{equation*}
\begin{split}
\phi^{\eps}_{(j_1\ldots j_k)}=\phi_{(j_1\ldots j_k)}\cdot \big(1-\eta(N^{x_0}/ \eps^{20})\big)=\phi_{(j_1\ldots j_k)}\cdot \widetilde{\eta}_\eps.
\end{split}
\end{equation*}
Clearly, $\phi^{\eps}_{(j_1\ldots j_k)}\in C^\infty _0(B_{\eps^{10}}(x_0))$ and
\begin{equation*}
\begin{cases}
&\square_\g\phi^{\eps}_{(j_1\ldots j_k)}=\widetilde{\eta}_{\eps}\cdot\square_\g\phi_{(j_1\ldots j_k)}+2\D_\al\phi_{(j_1\ldots j_k)}\cdot \D^\al \widetilde{\eta}_{\eps}+\phi_{(j_1\ldots j_k)}\cdot \square_\g\widetilde{\eta}_{\eps}\\
&\xi(\phi^{\eps}_{(j_1\ldots j_k)})=\widetilde{\eta}_{\eps}\cdot \xi(\phi_{(j_1\ldots j_k)})+\phi_{(j_1\ldots j_k)}\cdot \xi(\widetilde{\eta}_{\eps}).
\end{cases}
\end{equation*}
Using the Carleman inequality \eqref{tr11}, for any $(j_1,\ldots j_k)\in\{1,2,3,4\}^k$ we have
\begin{equation}\label{tr30}
\begin{split}
&\lambda\cdot \|e^{-\lambda \f_{\eps}}\cdot \widetilde{\eta}_{\eps}\phi_{(j_1\ldots j_k)} \|_{L^2}+\|e^{-\lambda \f_{\eps}}\cdot \widetilde{\eta}_{\eps}|D^1\phi_{(j_1\ldots j_k)}|\, \|_{L^2}\\
&\leq \widetilde{C}\lambda ^{-1/2}\cdot \|e^{-\lambda \f_{\eps}}\cdot \widetilde{\eta}_{\eps}\square_\g\phi_{(j_1\ldots j_k)}\|_{L^2}+\widetilde{C}\|e^{-\lambda \f_\eps}\cdot \widetilde{\eta}_\eps\xi(\phi_{(j_1\ldots j_k)})\|_{L^2}\\
&+\widetilde{C}\Big[\|e^{-\lambda \f_{\eps}}\cdot \D_\al\phi_{(j_1\ldots j_k)}\D^\al \widetilde{\eta}_{\eps} \|_{L^2}+\|e^{-\lambda \f_{\eps}}\cdot \phi_{(j_1\ldots j_k)} ( |\square_\g\widetilde{\eta}_{\eps}|+|D^1\widetilde{\eta}_{\eps}| )\|_{L^2}\Big],
\end{split}
\end{equation}
for any $\lambda\geq\widetilde{C}$. Using the identities in \eqref{va4}, in $B_{\eps^{10}}(x_0)$ we estimate pointwise
\begin{equation}\label{tr31}
\begin{cases}
&|\square_\g\phi_{(j_1\ldots j_k)}|\leq \widetilde{C}_{\Aa,\Bb,\Cc}\sum_{l_1,\ldots,l_k}\big(|D^1\phi_{(l_1\ldots l_k)}|+|\phi_{(l_1\ldots l_k)}|\big);\\
&|\xi(\phi_{(j_1\ldots j_k)})|\leq \widetilde{C}_{\Aa,\Bb,\Cc}\sum_{l_1,\ldots,l_k}|\phi_{(l_1\ldots l_k)}|,
\end{cases}
\end{equation}
for some constant $\widetilde{C}_{\Aa,\Bb,\Cc}$ that depends only on the constants $\widetilde{C}$ and the tensors $\Aa,\Bb,\Cc$. We add up the inequalities \eqref{tr30} over $(j_1,\ldots,j_k)\in\{1,2,3,4\}^k$. The key observation is that, in view of \eqref{tr31}, the first two terms in the right-hand side can be absorbed into the left-hand side for $\lambda$ sufficiently large. Thus, for any $\lambda\geq \widetilde{C}_{\Aa,\Bb,\Cc}$
\begin{equation}\label{tr32}
\begin{split}
\lambda&\sum_{j_1,\ldots ,j_k}\|e^{-\lambda \f_{\eps}}\cdot \widetilde{\eta}_{\eps}\phi_{(j_1\ldots j_k)} \|_{L^2}\\
&\leq \widetilde{C}\sum_{j_1,\ldots ,j_k}\Big[\|e^{-\lambda \f_{\eps}}\cdot \D_\al\phi_{(j_1\ldots j_k)}\D^\al \widetilde{\eta}_{\eps} \|_{L^2}+\|e^{-\lambda \f_{\eps}}\cdot \phi_{(j_1\ldots j_k)} ( |\square_\g\widetilde{\eta}_{\eps}|+|D^1\widetilde{\eta}_{\eps}| )\|_{L^2}\Big].
\end{split}
\end{equation}

Using the hypothesis \eqref{tr70} and the definition of the function $\widetilde{\eta}_\eps$, we have
\begin{equation*}
|\D_\al\phi_{(j_1\ldots j_k)}\D^\al \widetilde{\eta}_{\eps}|+\phi_{(j_1\ldots j_k)} ( |\square_\g\widetilde{\eta}_{\eps}|+|D^1\widetilde{\eta}_{\eps}| )\leq\widetilde{C}_{\phi}\cdot \mathbf{1}_{\{x\in B_{\eps^{10}}(x_0):\,r\geq R_0\text{ and }N^{x_0}>\eps^{20}/2\}},
\end{equation*}
for some $\widetilde{C}_{\phi}$ that depends on the smooth functions $\phi_{j_1\ldots j_k}$. Using the definition \eqref{tr12}, we observe also that
\begin{equation*}
\inf_{B_{\eps^{40}}(x_0)}e^{-\lambda \f_\eps}\geq e^{-\lambda \ln(\eps+\eps^{32}/2)}\geq\sup_{\{x\in B_{\eps^{10}}(x_0):\,r\geq R_0\text{ and }N^{x_0}>\eps^{20}/2\}}e^{-\lambda \f_\eps}.  
\end{equation*}
The identity \eqref{Tar3} follows by letting $\lambda\to\infty$ in \eqref{tr32}.

The set
\begin{equation*}
\{x\in\mathbf{E}^4:t(x)=0\text{ and }r(x)=R_0\}
\end{equation*}
is compact, where $t:\mathbf{E}^4\to\mathbb{R}$ is a smooth function which agrees with coordinate function $t$ in the Boyer-Lindquist coordinates. It follows from \eqref{tr10} that there is $\eps_3>0$ such that
\begin{equation}\label{tr80}
\Ss\equiv 0\text{ in the set }\{x\in\mathbf{E}^4:t(x)=0\text{ and }r(x)<R_0+\eps_3\}.
\end{equation}
We define the vectors $\widetilde{\partial}_1=\partial_r,\widetilde{\partial}_2=\partial_t,\widetilde{\partial}_3=\partial_\theta,\widetilde{\partial}_4=\partial_\phi\in\mathbb{T}(\widetilde{\mathbf{E}}^4)$ and the functions $\widetilde{\phi}_{(j_1\ldots j_k)}=\Ss(\widetilde{\partial}_{j_1},\ldots,\widetilde{\partial}_{j_k})$ as in the proof of Lemma \ref{Lemmaa2}. It follows from the identity \eqref{tr90} and \eqref{tr80} that
\begin{equation*}
\widetilde{\phi}_{(j_1\ldots j_k)}(r,t,\theta,\phi)=0\text{ if }r<R_0+\eps_3,
\end{equation*} 
which completes the proof of the lemma.
\end{proof}

\section{Proof of Theorem \ref{Main1}}\label{MinProof}

In this section we prove Theorem \ref{Main1}. We define the smooth optical functions $u,v:\mathcal{E}'\to(-1/2,\infty)$,
\begin{equation}\label{la15}
\begin{cases}
&u(t,x)=|x|-1-t;\\
&v(t,x)=|x|-1+t,
\end{cases}
\end{equation}
where $\mathcal{E}'=\{(t,x)\in\mathcal{M}:|x|>|t|+1/2\}$. Notice that $\mathcal{E}=\{(t,x)\in\mathcal{E}':u>0\text{ and }v>0\}$. For $R\in[1,\infty)$ we define the relatively compact open set
\begin{equation}\label{ca80}
\mathcal{E}_{R}=\{(t,x)\in\mathcal{E}:(u+1/2)(v+1/2)<R\}.
\end{equation}

\begin{proposition}\label{Carmin1}
Assume $R\geq 1$. Then there is $\lambda (R)\gg 1$ such that for any $\phi\in C^2_0(\mathcal{E}_{R})$ and $\lambda\geq\lambda(R)$
\begin{equation}\label{ca1}
\lambda \cdot \|e^{-\lambda f}\cdot \phi\|_{L^2}+
\|e^{-\lambda f}\cdot D \phi\|_{L^2}\leq C_R\lambda^{-1/2}\cdot \|e^{-\lambda f}\cdot\square \phi\|_{L^2},
\end{equation}
where
\begin{equation}\label{ca6}
f=\log(u+1/2)+\log(v+1/2)=\log\big[( |x|-1/2)^2-t^2\big].
\end{equation}
and
 $
|D \phi |=\big( \sum_{\mu=0}^d|\partial_\mu \phi |^2\big)^{1/2}
$.
\end{proposition}

The Carleman inequality in Proposition \ref{Carmin1} suffices to prove Theorem \ref{Main1}, by an argument similar to the one given in Lemma \ref{Lemmaa1} (which exploits implicitly the bifurcate characteristic geometry of $\mathcal{H}$, using a cutoff function of the form $\eta(uv/\delta)$, to compensate for the fact that we do not assume vanishing of the derivatives of $\phi$ on $\mathcal{H}$). Proposition \ref{Carmin1} can  be obtained as a direct consequence of H\"{o}rmander's general pseudo-convexity condition \eqref{HoCond}. For the convenience of the reader, we provide below a self-contained elementary proof of Proposition \ref{Carmin1}, in which we verify implicitly a similar pseudo-convexity condition in our simple case and show how it implies the Carleman inequality. 

\begin{proof}[Proof of Proposition \ref{Carmin1}] The constants $C\geq 1$ in this proof may depend on $R$ and $d$. We may assume that $\phi\in C^\infty_0(\mathcal{E}_{R})$ is real-valued. Since   all  partial derivatives  of $f$ are  bounded  in $\mathcal{E}_{R}$, for \eqref{ca1} it suffices to prove that, for $\lambda\geq\lambda(R)$,
\begin{equation}\label{car1}
\lambda \cdot \|e^{-\lambda f}\cdot \phi\|_{L^2}+\| D(e^{-\lambda  f}\cdot \phi\big)\|_{L^2}\leq C\lambda ^{-1/2}\cdot \|e^{-\lambda  f}\cdot\square \phi\|_{L^2}.
\end{equation}
To prove  estimate \eqref{car1} we start by setting,
\bea
\phi=e^{\la f}\psi
\eea
with $f=f(u,v)$  as above. 
Observe that,
\beaa
e^{-\la f}\square(e^{\la f}\psi)=\square\psi+\la(2 \D^\b f \D_\b \psi+\square f\psi)+\la^2( \D^\be f \D_\b f)\psi.
\eeaa
Thus estimate  \eqref{car1}
follows from,
\beaa
\la\|\psi\|_{L^2}+ C^{-1}\|D\psi\|_{L^2}\le C \la^{-1/2}
 \|L\psi+\la (\square f)\psi \|_{L^2},
\eeaa
where, 
\begin{equation*}
\begin{split}
&L\psi=\square\psi+ 2\la W\psi +\la^2 G\psi,\\
&W=\D^\a f\D_\a,\qquad G=\D^\be f \D_\b f.
\end{split}
\end{equation*}
Since $\square f$ 
is bounded on $\EE_R$,  i.e. $|\square  f|\le C$, it suffices in fact to
 show that,
 \bea
\la\|\psi\|_{L^2}+ C^{-1} \|D \psi\|_{L^2}\le C \la^{-1/2}.
 \|L\psi \|_{L^2}\label{car2}
\eea
We shall establish in fact a lower bound for 
an integral of the  form,
\bea
E=<L\psi, 2\la( W-w)\psi>=2\la \int_{\EE_R}L\psi \big(W(\psi)-w\psi\big) 
\eea
where $w$ is a smooth function on $\EE_R$ we will choose below. In fact
we will choose $w$ such that we can establish the lower bound,
\bea
E\ge C^{-1}\big(\la\| D\psi\|_{L^2}^2+\la^3\|\psi\|_{L^2}^2\big)+\la^2 \|(W-w)\psi  \|_{L^2}^2   \label{car3}
\eea
Since $E \le  \|L\psi\|^2_{L^2}+\lambda^2\|(W-w)\psi\|^2_{L^2}$ \eqref{car2} easily follows from \eqref{car3}.

Now, writing $L\psi=\square\psi+\lambda^2G\psi+\lambda(W\psi+w\psi)+\lambda(W\psi-w\psi)$,
\begin{equation}\label{car5}
\begin{split}
&E=2\la<L\psi, (W-w)\psi>=2\la^2\|(W-w)\psi\|_{L^2}^2
+2\lambda^2\|W\psi\|_{L^2}^2-2\lambda^2\|w\psi\|_{L^2}^2+E_1+E_2\\
&E_1=\la<\square\psi, (2W-2w)\psi>\\
&E_2=\lambda^3<G\psi,     (2W-2w)\psi>.
\end{split}
\end{equation}
Thus,  for bounded $w$   and   for $\lambda$ sufficiently large,  \eqref{car3} is an immediate consequence
of 
\bea
2\lambda^2\|W\psi\|_{L^2}^2+E_1+E_2 \ge  C^{-1}\big(\la\| D\psi\|_{L^2}^2+\la^3\|\psi\|_{L^2}^2\big),\label{carr4}
\eea
 To evaluate $E_1$ and $E_2$ we  make use of the following simple lemma.

\begin{lemma}
 Let $Q_{\a\be}=\D_\a\psi \D_\b\psi-\frac{1}{2}m_{\a\b}( \D^\mu\psi  \D_\mu \psi)$ denote the enery-momentum tensor of the   wave operator 
    $\square =m^{\a\b}  \D_\a  \D_\b$. Then,    
\beaa
\square\psi\cdot(2W\psi-2 w\psi)&=&\D^\al(2W^\be Q_{\al\be}-2w\psi\cdot D_\al\psi+ D_\al w\cdot \psi^2)\\
&-&Q^{\a\b} (\D_\al W_\be+ \D_\b W_\a)+2w \D^\al\psi\cdot  \D_\al\psi-\square_{\g}w\cdot \psi^2,
\eeaa
and
\beaa
G\psi\cdot(2W\psi-2 w\psi)&=&\D^\al(\psi^2G\cdot W_\al)-\psi^2(2wG+W(G)+G\cdot \D^\al W_\al).
\eeaa
\end{lemma}

Since $\psi\in C^\infty_0(\EE_R)$ we integrate by parts to conclude that
\begin{equation}
\begin{split}
E_1+E_2&=\lambda\int_{\EE_R}2w\D^\al\psi\cdot \D_\al\psi-2\D^\al W^\be\cdot Q_{\al\be}\\
&+\lambda ^3\int_{\EE_R}\psi^2(-2wG-W(G)-G\cdot \D^\al W_\al)\\
 &-\lambda \int_{\EE_R}\psi^2\square_\g w.
\end{split}
\label{carr44}
\end{equation}
To prove  \eqref{carr4}    we are reduced to prove    pointwise bounds  for the first two integrands 
  in  \eqref{carr44}.  More precisely,
  dividing by $\lambda$ and $\lambda^3$ respectively,  it suffices to prove that the pointwise bounds
\begin{equation}\label{Car20}
C^{-1}|D\psi|^2\leq \lambda|W(\psi)|^2+(w\D^\al\psi\cdot \D_\al\psi-\D^\al W^\be\cdot Q_{\al\be}),
\end{equation}
and
\begin{equation}\label{Car21}
C^{-1}\leq-2wG-W(G)-G\cdot \D^\al W_\al,
\end{equation}
hold on $\EE_R$, for $\lambda$ sufficiently large. 

Recall that $W^\al=\D^\al f$ and $G=\D_\al f\D^\al f$. Observe that
\begin{equation*}
w\D^\al\psi\cdot \D_\al\psi-\D^\al W^\be\cdot Q_{\al\be}=(\D^\al\psi\cdot \D^\be\psi)[(w+\square f/2)m_{\al\be}-\D_\al\D_\be f]
\end{equation*}
and
\begin{equation*}
-2wG-W(G)-G\cdot \D^\al W_\al=-G(2w+\square f)-2\D^\al f\D^\be f \cdot \D_\al \D_\be f.
\end{equation*}
Thus, with $w'=w+\square f/2\in C^\infty(\EE_R)$ (still to be chosen), the inequalities \eqref{Car20} and \eqref{Car21} are equivalent to the pointwise inequalities
\begin{equation}\label{Car22}
C^{-1}|D\psi|^2\leq\lambda|\D_\alpha f\cdot\D^\alpha\psi|^2+ (\D^\al\psi\cdot \D^\be\psi)(w'm_{\al\be}-\D_\al\D_\be f),
\end{equation}
and
\begin{equation}\label{Car23}
C^{-1}\leq-w' (\D_\a f \D^\a f)-\D^\al f\,\D^\be f\cdot \D_\al \D_\be f
\end{equation}
on $\EE_R$, for $\lambda$ sufficiently large.

Let $h= e^f$ or, in view of \eqref{ca6},   $h= (|x|-1/2)^2-t^2)$. In terms of $h$ making use 
 of the inequality $h\ge 1/4$ on $\EE_R$, the inequalities \eqref{Car22} and \eqref{Car23} are equivalent to
\begin{equation}\label{Car24}
C^{-1}|D\psi|^2\leq \lambda |\D_\alpha h\cdot\D^\alpha\psi|^2+ (\D^\al\psi\cdot \D^\be\psi)(w'm_{\al\be}-h^{-1}\D_\al\D_\be h),
\end{equation}
and
\begin{equation}\label{Car25}
C^{-1}\leq \D^\al h\D^\be h(h^{-2}\D_\al h\D_\be h-h^{-1}\D_\al\D_\be h)-w'\D_\al h\D^\al h,
\end{equation} 
provided that $\lambda$ is sufficiently large. To summarize, we need to find $w'\in C^\infty(\EE_R)$ such that the inequalities \eqref{Car24} and \eqref{Car25} hold in $\EE_R$, for all $\lambda$ sufficiently large.

We shall see below that  our function  $h$, strictly positive and smooth on $\EE_R$ verifies the equation,
\bea
\D^\a h \D_\a h =4h\label{eikonal-h}
\eea
We  infer  by differentiation that,
$
\D_{\a}\D_{\b} h \, \D^\b h=2 \D_\a h
$
and therefore,
\beaa
\D_{\a}\D_{\b} h\,\D^\a h\,  \D^\b h= 8h.
\eeaa
Therefore the right-hand side of \eqref{Car25} is equal to $8-4hw'$
 and thus inequality \eqref{Car25}   is equivalent to 
 $
 h w'\le 2-C^{-1}\,\, \mbox{in}\,\, \EE_R,
 $
 which is clearly satisfied   if 
\begin{equation}\label{fixw}
w'=h^{-1}(2-A_0|x|^{-1})\quad\text{ for some constant }A_0>0.
\end{equation}

 On the other hand, setting $Y^\al=\D^\al\psi$
 and $H_{\a\b}=\D_\a\D_\b h$,  $\al,\b=0,\ldots,d$
 and  observing that  $H_{0i}=0$ for $i=1,\ldots, d$,
 we infer that
 the right-hand side of \eqref{Car24} is equal to,
 \beaa
E:&=& \la (\D_\a h Y^\a)^2+ w'\big( -(Y^0)^2 +|Y'|^2)- h^{-1}\big( H_{00}( Y^0)^2
 +H_{ij} Y^i Y^j\big)\\
 &=&(Y^0)^2\big(-w'-h^{-1} H_{00}\big)+|Y'|^2\big(w'+ H_{ij} \hat{Y}^i \hat{Y}^j\big)+
  \la (\D_0 h Y^0 +\D_i h\, Y^i )^2
 \eeaa
 where $|Y'|^2=\sum_{i=1}^d (Y^i)^2$ and $\hat Y^i =|Y'|^{-1} Y^i$.
 Since
$
h=(|x|-1/2)^2-t^2,
$
we have $|h|+|h^{-1}|+|x|+(|x|-1/2)^{-1}\leq C$ in $\EE_R$. We compute
\begin{equation}\label{der1}
\D_0h=-2t,\qquad \D_jh=(2-|x|^{-1})x_j\text{ for }j=1,\ldots,d,
\end{equation}
and
\begin{equation}\label{der2}
\begin{split}
&H_{00}=\D_0\D_0h=-2\\
&H_{ij}=\D_i\D_jh=(2-|x|^{-1})\delta_{ij}+x_ix_j|x|^{-3}\text{ for }i,j=1,\ldots,d.
\end{split}
\end{equation}
Thus we easily check that \eqref{eikonal-h} is  indeed verified. 
 Setting  $Z=Y\c \hat{x}$, with $\hat x_i= \frac{x_i}{|x|}$,  the expression for $E$ becomes,
\beaa
  E&=&(Y^0)^2h^{-1} \big(2  - h w')+|Y'|^2 h^{-1}(hw'-(2-|x|^{-1}) -h^{-1}|x|^{-1}  Z^2\\
 & +&\la\big(-2t Y^0+ (2|x|-|1) Z\big)^2\\
 &=& h^{-1} A_0 |x|^{-1} (Y^0)^2+h^{-1}(1-A_0)|x|^{-1}|Y'|^2-h^{-1}|x|^{-1}  Z^2+\la\big(-2t Y^0+ (2|x|-|1) Z\big)^2
\eeaa
To  derive the bound,
\bea
E\ge C^{-1}\big((Y^0)^2+|Y'|^2\big),\label{Car24-again}
\eea
  from which\eqref{Car24}  follows,
  we rely on the following simple lemma.
\begin{lemma}
\label{le.ineq-final}
Given  $\delta>0$ there  exists $\lambda$ sufficiently large  (depending on $R$ and $\delta$) such that the following inequality holds:
\begin{equation}
\begin{split}
\lambda\Big[(2|x|-1)Z-2tY^0\Big]^2&+h^{-1}A_0|x|^{-1}(Y^0)^2-h^{-1}|x|^{-1}Z^2\\
&\geq (Y^0)^2h^{-1}|x|^{-1}\Big(A_0-\frac{t^2}{(|x|-1/2)^2}-\delta\Big),
\end{split}
\label{ineq-final}
\end{equation}
\end{lemma}
In view of the lemma the bound  \eqref{Car24-again}  follows by choosing $A_0=1-C_0^{-1}$ and $\delta=C_0^{-1}$, for $C_0$ sufficiently large depending on $R$. This completes the proof of the proposition.
\end{proof}
We  give   below the proof of Lemma \ref{le.ineq-final}.
\begin{proof}
Inequality  \eqref{ineq-final}  is equivalent to,
\beaa
\lambda (2|x|-1)^2\Big[Z-\frac{t}{(|x|-1/2)}Y^0\Big]^2+
h^{1}|x|^{-1}\big(( Y^0 \frac{t}{|x|-1/2})^2-Z^2 \big) +\de h^{-1}|x|^{-1} (Y^0)^2\ge 0
\eeaa
Setting $X=\frac{t}{|x|-1/2}Y^0-Z$ we  can rewrite  the above inequality  in the form,
\beaa
\lambda (2|x|-1)^2 X^2+h^{-1}|x|^{-1}X (-X+ 2\frac{t}{|x|-1/2}) Y^0)+\de h^{-1}|x|^{-1} (Y^0)^2\ge 0
\eeaa
or, equivalently,
\beaa
X^2\big(\lambda (2|x|-1)^2-h^{-1}|x|^{-1}\big)+2\frac{t}{|x|-1/2}X Y^0
+\de h^{-1}|x|^{-1} (Y^0)^2\ge 0
\eeaa
which clearly holds for $t, x$ in $\EE_R$ and all $X, Y^0$ in $\RR$
provided that $\la$ is sufficiently large.
\end{proof}

\appendix

\section{Explicit computations in the Kerr spaces}\label{explicit}

We consider the exterior region $\mathbf{E}^4$ of the Kerr spacetime of mass $m$ and angular momentum $ma$, $a\in[0,m)$. Following \cite[Chapter 6]{Ch}, in the standard Boyer-Lindquist coordinates $(r,t,\theta,\phi)\in (r_+,\infty)\times\mathbb{R}\times(0,\pi)\times\mathbb{S}^1$, $r_\pm=m\pm (m^2-a^2)^{1/2}$, the Kerr metric on a dense open subset $\widetilde{\mathbf{E}}^4$ of $\mathbf{E}^4$ is
\begin{equation}\label{k1}
ds^2=-\frac{\rho^2\Delta}{\Sigma^2}(dt)^2+\frac{\Sigma^2(\sin\theta)^2}{\rho^2}\Big(d\phi-\frac{2amr}{\Sigma^2}dt\Big)^2+\frac{\rh^2}{\Delta}(dr)^2+\rho^2(d\theta)^2,
\end{equation}
where
\begin{equation}\label{k2}
\begin{cases}
&\Delta=r^2+a^2-2mr;\\
&\rho^2=r^2+a^2(\cos\theta)^2;\\
&\Sigma^2=(r^2+a^2)\rho^2+2mra^2(\sin\theta)^2=(r^2+a^2)^2-a^2(\sin\theta)^2\Delta.
\end{cases}
\end{equation}
This metric is of the form
\begin{equation}\label{g1}
ds^2=-e^{2\nu}(dt)^2+e^{2\psi}(d\phi-\omega dt)^2+e^{2\mu_2}(dr)^2+e^{2\mu_3}(d\theta)^2,
\end{equation}
where
\begin{equation}\label{k3}
\begin{split}
&e^{2\nu}=\frac{\rho^2\Delta}{\Sigma^2}\text{ and }\nu=\frac{1}{2}[\ln(\rho^2)+\ln\Delta-\ln(\Sigma^2)]\\
&e^{2\psi}=\frac{\Sigma^2(\sin\theta)^2}{\rho^2}\text{ and  }\psi =\frac{1}{2}[\ln(\Sigma^2)+2\ln(\sin\theta)-\ln(\rho^2)];\\
&\omega=\frac{2amr}{\Sigma^2};\\
&e^{2\mu_2}=\frac{\rh^2}{\Delta}\text{ and }\mu_2=\frac{1}{2}[\ln(\rho^2)-\ln\Delta];\\
&e^{2\mu_3}=\rho^2\text{ and }\mu_3=\frac{1}{2}\ln (\rho^2).
\end{split}
\end{equation}
We compute
\begin{equation}\label{k4}
\begin{split}
&\partial_r\mu_2=\frac{r}{\rho^2}-\frac{r-m}{\Delta}\,\text{ and }\,\partial_\theta \mu_2=\frac{-a^2\sin\theta\cos\theta}{\rho^2};\\
&\partial_r\mu_3=\frac{r}{\rho^2}\,\text{ and }\,\partial_\theta\mu_3=\frac{-a^2\sin\theta\cos\theta}{\rho^2};\\
\end{split}
\end{equation}
and
\begin{equation}\label{k4.6}
\begin{split}
&\partial_r\omega=-\frac{2am}{\Sigma^4}[(3r^2-a^2)(r^2+a^2)-a^2(\sin\theta)^2(r^2-a^2)];\\
&\partial_\theta\omega=\frac{4a^3mr\Delta \sin\theta\cos\theta}{\Sigma^4};\\
&\partial_r\nu=\frac{r}{\rho^2}+\frac{r-m}{\Delta}-\frac{2r(r^2+a^2)-a^2(\sin\theta)^2(r-m)}{\Sigma^2};\\
&\partial_\theta\nu=a^2\sin\theta\cos\theta\big(\frac{\Delta}{\Sigma^2}-\frac{1}{\rho^2}\big);\\
&\partial_r\psi=\frac{2r(r^2+a^2)-a^2(\sin\theta)^2(r-m)}{\Sigma^2}-\frac{r}{\rho^2};\\
&\partial_\theta \psi=-a^2\sin\theta\cos\theta\big(\frac{\Delta}{\Sigma^2}-\frac{1}{\rho^2}\big)+\frac{\cos\theta}{\sin\theta}.
\end{split}
\end{equation}
We fix the frame
\begin{equation}\label{g3}
e_0=e^{-\nu}(\partial_t+\omega\partial_\phi),\,\,e_1=e^{-\psi}\partial_\phi,\,\,e_2=e^{-\mu_2}\partial_r,\,\,e_3=e^{-\mu_3}\partial_\theta.
\end{equation}
Clearly, ${\g}_{\al\be})=({\g}^{\al\be})=\mathrm{diag}(-1,1,1,1)$, where ${\g}_{\al\be}={\g}(e_\al,e_\be)$, $\al,\be=0,1,2,3$. The dual basis of $1$-forms is
\begin{equation}\label{g3.1}
\eta^0=e^\nu dt,\,\,\eta^1=e^\psi(d\phi-\omega dt),\,\,\eta^2=e^{\mu_2}dr,\,\,\eta^3=e^{\mu_3}d\theta.
\end{equation}
Also
\begin{equation}\label{g10}
\xi=\partial_t=e^\nu\cdot e_0-e^\psi \omega\cdot e_1.
\end{equation}

We compute now the covariant derivatives ${\D}_{e_i}e_j$, $i,j=0,1,2,3$. We use the formula
\begin{equation}\label{g5}
\begin{split}
{\g}(Z,{\D}_YX)=&\frac{1}{2}\big(X({\g}(Y,Z))+Y({\g}(Z,X))-Z({\g}(X,Y))\\
&-{\g}([X,Z],Y)-{\g}([Y,Z],X)-{\g}([X,Y],Z)\big),
\end{split}
\end{equation}
for any vector fields $X,Y,Z$. We have
\begin{equation}\label{g6}
\begin{split}
&[e_0,e_1]=0;\\
&[e_0,e_2]=e^{-\mu_2}\partial_r\nu\cdot e_0-e^{\psi-\mu_2-\nu}\partial_r\omega\cdot e_1;\\
&[e_0,e_3]=e^{-\mu_3}\partial_\theta\nu\cdot e_0-e^{\psi-\mu_3-\nu}\partial_\theta\omega\cdot e_1;\\
&[e_1,e_2]=e^{-\mu_2}\partial_r\psi\cdot e_1;\\
&[e_1,e_3]=e^{-\mu_3}\partial_\theta\psi\cdot e_1;\\
&[e_2,e_3]=e^{-\mu_3}\partial_\theta\mu_2\cdot e_2-e^{-\mu_2}\partial_r\mu_3\cdot e_3.\\
\end{split}
\end{equation}
With $[e_i,e_j]=C^k_{ij}e_k$, $C_{ij}^k+C_{ji}^k=0$, it follows from \eqref{g5}
that
\begin{equation}\label{g7.3}
\D_{e_j}e_i=-\frac{1}{2}\sum_{k=0}^3(\g_{jj}\g_{kk}C_{ik}^j+\g_{ii}\g_{kk}C_{jk}^i+C_{ij}^k)e_k.
\end{equation}
Using the table \eqref{g6}, this gives
\begin{equation}\label{g8}
\begin{split}
&\D_{e_0}e_0=C_{02}^0e_2+C_{03}^0e_3;\,\,\D_{e_1}e_0=\frac{-1}{2}C_{02}^1e_2+\frac{-1}{2}C_{03}^1e_3;\\
&\D_{e_2}e_0=\frac{-1}{2}C_{02}^1e_1;\,\,\D_{e_3}e_0=\frac{-1}{2}C_{03}^1e_1;\\
&\D_{e_0}e_1=\frac{-1}{2}C_{02}^1e_2+\frac{-1}{2}C_{03}^1e_3;\,\,\D_{e_1}e_1=(-1)C_{12}^1e_2+(-1)C_{13}^1e_3\\
&\D_{e_2}e_1=\frac{-1}{2}C_{02}^1e_0;\,\,\D_{e_3}e_1=\frac{-1}{2}C_{03}^1e_0;\\
&\D_{e_0}e_2=C_{02}^0e_0+\frac{1}{2}C_{02}^1e_1;\,\,\D_{e_1}e_2=\frac{-1}{2}C_{02}^1e_0+C_{12}^1e_1;\\
&\D_{e_2}e_2=-C_{23}^2e_3;\,\,\D_{e_3}e_2=-C_{23}^3e_3;\\
&\D_{e_0}e_3=C_{03}^0e_0+\frac{1}{2}C_{03}^1e_1;\,\,\D_{e_1}e_3=\frac{-1}{2}C_{03}^1e_0+C_{13}^1e_1;\\
&\D_{e_2}e_3=C_{23}^2e_2;\,\,\D_{e_3}e_3=C_{23}^3e_2.
\end{split}
\end{equation}

Let
\begin{equation*}
Y=2r(r^2+a^2)-a^2(\sin\theta)^2(r-m),
\end{equation*}
and observe that
\begin{equation}\label{de5}
(3r^2-a^2)(r^2+a^2)-a^2(\sin\theta)^2(r^2-a^2)=2rY-\Sigma^2>0,
\end{equation}
and
\begin{equation}\label{de6}
2r\Sigma^2>\rho^2 Y.
\end{equation}

We compute now the Hessian $\D^2 r$. More generally, for a function $f$ that depends only on $r$ (i.e. $e_0(f)=e_1(f)=e_3(f)=0$), using \eqref{g6} and \eqref{g8}, and the formula $\D_\al\D_\be f=\D_\be\D_\al f=e_\alpha(e_\beta(f))-\D_{e_\al}e_\be(f)$, 
\begin{equation}\label{de1}
\begin{split}
&\D_0\D_0 f=-C_{02}^0e^{-\mu_2}\partial_rf=-\frac{\Delta}{\rho^2}\Big(\frac{r}{\rho^2}+\frac{r-m}{\Delta}-\frac{Y}{\Sigma^2}\Big)\partial_rf\\
&\D_0\D_1 f=\frac{1}{2}C_{02}^1e^{-\mu_2}\partial_rf=\frac{\Delta}{\rho^2}\cdot \frac{ma\sin\theta}{\rho^2\sqrt{\Delta}\Sigma^2}(2rY-\Sigma^2)\partial_rf\\
&\D_1\D_1 f=C_{12}^1e^{-\mu_2}\partial_rf=\frac{\Delta}{\rho^2}\Big(\frac{Y}{\Sigma^2}-\frac{r}{\rho^2}\Big)\partial_rf\\
&\D_2\D_2 f=e^{-\mu_2}\partial_r(e^{-\mu_2}\partial_rf)=\frac{\Delta}{\rho^2}\partial^2_rf-\frac{\Delta}{\rho^2}\Big(\frac{r}{\rho^2}-\frac{r-m}{\Delta}\Big)\partial_rf\\
&\D_2\D_3 f=-C_{23}^2e^{-\mu_2}\partial_rf=\frac{\sqrt{\Delta}a^2\sin\theta\cos\theta}{\rho^4}\partial_rf\\
&\D_3\D_3 f=-C_{23}^3e^{-\mu_2}\partial_rf=\frac{\Delta r}{\rho^4}\partial_rf\\
&\D_0\D_2 f=\D_0\D_3 f=\D_1\D_2 f=\D_1\D_3 f=0.
\end{split}
\end{equation}

\end{document}